\documentclass{article}
\usepackage[utf8]{inputenc}
\usepackage[dvipsnames]{xcolor}
\usepackage[colorlinks = true,
            linkcolor = blue,
            urlcolor  = blue,
            citecolor = blue,
            anchorcolor = blue
            hydelinks]{hyperref}
    
\usepackage[english]{babel}
\usepackage{graphicx}
\usepackage[graphicx]{realboxes}
\usepackage{wrapfig}
\usepackage{caption}
\usepackage{subcaption}
\usepackage{algpseudocode}
\usepackage{amsfonts}
\usepackage{multirow}
\usepackage{rotating}
\usepackage{float}  
\usepackage{mathtools}
\usepackage{longtable}
\graphicspath{ {figures/} }
\usepackage{array}
\usepackage{lscape}
\usepackage{epstopdf}
\usepackage{amsthm}
\usepackage{amsmath}
\usepackage{enumerate}
\usepackage{amssymb}
\usepackage{tabularx}
\usepackage{authblk}
\usepackage{comment}
\usepackage{a4wide}
\usepackage{adjustbox}
\usepackage{stackengine}
\makeatletter

\renewcommand*\env@matrix[1][*\c@MaxMatrixCols c]{%
  \hskip -\arraycolsep
  \let\@ifnextchar\new@ifnextchar
  \array{#1}}
\makeatother
\usepackage{systeme}

\usepackage[nottoc]{tocbibind}
\usepackage{csquotes}

\hypersetup{
	bookmarks=true,         
	unicode=false,          
	pdftoolbar=true,        
	pdfmenubar=true,        
	pdffitwindow=false,     
	pdfstartview={FitH},    
	pdftitle={My title},    
	pdfauthor={Author},     
	pdfsubject={Subject},   
	pdfcreator={Creator},   
	pdfproducer={Producer}, 
	pdfkeywords={keywords}, 
	pdfnewwindow=true,      
	colorlinks=true,       
	linkcolor=blue,          
	citecolor=blue,        
	filecolor=magenta,      
	urlcolor=cyan           
}

\newtheorem{lemma}{Lemma}

\newtheorem{algorithm}{Algorithm}
\newtheorem*{remark}{Remark}

\title{Cost of institutional incentives for promoting cooperation in 2×2 games and collective risk games}
\author[$\dagger$]{M. H. Duong}
\author[$\star$]{C. M. Durbac}
\author[$\ddag$]{T. A. Han}

\affil[$\dagger$]{School of Mathematics, University of Birmingham, UK. Email: h.duong@bham.ac.uk}
\affil[$\star$]{School of Mathematics, University of Birmingham, UK. Email: c.m.durbac@bham.ac.uk}
\affil[$\ddag$]{ School of Computing, Engineering and Digital Technologies, Teesside University, UK. Email: T.Han@tees.ac.uk}
\date\today

\begin{document}

\maketitle
\begin{abstract}
Prosocial behaviours have been extensively studied across multiple disciplines. Cooperation, requiring a personal cost for collective benefits, is widespread in nature and human society, having been explained through mechanisms such as kin selection, direct and indirect reciprocity, and network reciprocity. Institutional incentives, which reward cooperation and punish anti-social behaviour, offer a promising approach to fostering cooperation in groups of self-interested individuals. Focusing on general $2\times2$ games and the collective risk game (which is a fundamental model for climate action), we analyse the associated cost of providing incentives under evolutionary dynamics governed by Fermi's rule, exploring the asymptotic behaviour of the incentive cost functons in the limits of neutral drift and strong selection. We also implement numerical simulations to study how parameters such as the intensity of selection affect the behaviour of the aforementioned cost functions.
\end{abstract}
\newpage

\tableofcontents

\section{Introduction}

The understanding of what motivates prosocial behaviours in self-regarding groups has been the topic of debate in many disciplines such as Mathematics, Computer Science, Sociology, and Anthropology. Seemingly against Darwinian natural selection which emphasises payoff maximisation, cooperation involves incurring personal costs for collective benefits. However, cooperation is pervasive in nature, from bacteria to insects to human societies. Scholars have studied multiple mechanisms for promoting it, among which kin selection, direct and indirect reciprocity, network reciprocity \cite{nowak2006,sigmund2010calculus,perc2017statistical,rand2013human,van2014reward,xia2023reputation,hu2020rewarding,moralpref}, and, the subject of the work at hand, institutional incentives \cite{sasaki2012take,sigmundinstitutions,wang2019exploring,duong2021cost,cimpeanu2021cost,sun2021combination,van2014reward,gurerk,gois2019reward,sun2021combination,liu2022effects,flores2024evolution,wang2024partial,hua2024coevolutionary,han2024evolutionary,liu2025evolution}. 

Institutional incentives involve the existence of an external decision-maker who imposes rewards (which increase the payoff of cooperators), punishments (which decrease the payoff of defectors), or hybrid combinations of both. Determining the optimal design of such incentives remains a central research question, particularly in finite populations modelled on Markov chains where stochastic effects such as neutral drift and the intensity of selection significantly influence evolutionary outcomes. Previous studies have examined the role of incentives in Public Goods Games and the Collective Risk Game, particularly in scenarios such as climate action. These works demonstrated that well-calibrated incentives can stabilise cooperation while reducing unnecessary costs \cite{han2017evolution,gois2019reward}. Using optimal control theory, Wang et al. \cite{wang2019exploring} obtained the specific optimal control protocols for institutional reward and institutional punishment in the case of the Public Goods Game such that the cost per capita paid by the institution is minimised.

In the case of the Prisoner’s Dilemma (PD), although mutual cooperation is Pareto superior to mutual defection, defection strictly dominates cooperation. Previous studies have examined how evolutionary mechanisms, such as reputation, reciprocity, and spatial structure, can promote cooperative strategies \cite{nowak2006,rand2013human}. More recently, institutional interventions have been introduced to stabilise cooperation in finite populations. For instance, employing optimal control theory, Rui et al. \cite{rui2018optimal} showed that carefully chosen reward and punishment schemes can sustain high levels of cooperation. 
Both aforementioned works do not take into consideration stochastic effects such as the intensity of selection which are relevant when analysing evolutionary dynamics. In the paper at hand, the intensity of selection is a central parameter of our model, drastically affecting the behaviour of the incentive cost functions.

While in the Prisoner's Dilemma cooperation is costly to the player, in Stag Hunt (SH) it is not. This is because there is no temptation to deviate from prosocial action if the other player cooperates, however cooperation is risky since it pays off if and only there is cooperation on the part of the opponent. From a behavioural point of view, Capraro et al. \cite{capraro2020preferences} observed that cooperation in SH is driven more by efficiency than by morality or fairness, an interesting contrast to PD where cooperation is influenced by morality. The key difference between these two games (which are similar in the sense that in both there is a group benefit) is that, for one cooperation is costly (PD), while for the other it is not (SH). This, as noted in \cite{capraro2020preferences}, suggests that the element of morality is strongly tied to how costly it is for players to cooperate as opposed to how much it benefits group dynamics. Here, we are solely interested in the game's ability to convey the cooperative behaviours of players, not on the reason for cooperating.

The Hawk and Dove (HD) game, a model of conflict over scarce resources, poses a different challenge in terms of cooperative behaviour, its equilibrium dynamics often favouring a mix of aggressive (hawk) and passive (dove) strategies. If the cost of fighting is very high, a dovish behaviour is more desirable, while, if the resource value is very high, a hawkish strategy is more favourable. For the purpose of this paper, we assume that the cost of fighting is very high and thus assume that dove is the cooperative strategy. HD has not been extensively studied in the framework of institutional incentives, however it has been the centre of focus of experimental studies because of its ability to describe the dynamics underlying international relations, in particular, negotiations \cite{clare2014hawks, mattes2019hawks}. In structured populations under a dynamical network model, it has been shown that cooperation in HD is robust when individuals are allowed to break ties with undesirable neighbours and to create new ties in their extended neighbourhood \cite{tomassini2010mutual}. In this work, we analyse the incentive cost functions when the underlying interaction is a HD game in well-mixed, finite populations.

This paper addresses the problem of optimising the cost of institutional incentives provided by an external decision-maker to maximise, or at least ensure, a desired level of cooperation within a population of self-regarding individuals. We consider two games, general $2\times2$ games (with particular attention to cooperative and defective versions of the Prisoner's Dilemma) and the Collective Risk Game, commonly used to model climate action. Evolutionary dynamics follow Fermi's rule, allowing us to capture the effect of the intensity of selection on strategy updates. Although many simulation-based studies have explored this question, rigorous analytical works remain limited \cite{wang2019exploring, han2018cost, duong2021cost, DuongDurbacHan2022, wangdecentralised, wang2023optimization}. This paper contributes to closing that gap by analysing the behaviour of reward, punishment, and hybrid incentive cost functions in both neutral drift and strong selection limits, in the case of games in which the average payoff difference between cooperators and defectors depends on the population composition. Additionally, we use numerical simulations to investigate the behaviour of cost functions across a range of parameters, assuming a general starting state, dependant on the number of cooperators and defectors computed at equilibrium.

\paragraph{Overview of contribution of this paper.}

Building on the discrete approach established in \cite{han2018cost, duong2021cost, DuongDurbacHan2022, DuongDurbacHan2024}, we investigate the problem of optimising the cost of institutional incentives - specifically reward, punishment, and hybrid schemes - to maximise cooperative behaviour (or ensure a minimum level of cooperation) in well-mixed, finite populations. We focus on a full-incentive scheme, where every player receives incentives in each generation.

The general setting, for a population of size $N$, introduces several layers of complexity due to the number of parameters involved: population size, intensity of selection, game-specific parameters, and the relative efficiency of each incentive scheme. When evolutionary dynamics are modelled as a Markov chain of order $N$, as is the case here, the fundamental matrix - central to calculating expected incentive costs - becomes increasingly intractable to evaluate both analytically and computationally. To address this, we supplement our theoretical results with numerical simulations, studying how the incentive cost functions behave under different assumptions and parameter regimes.

Whereas our previous works \cite{DuongDurbacHan2022,DuongDurbacHan2024} have focused on settings in which the payoff difference between cooperators and defectors is independent of the state of the Markov chain (counting the overall number of cooperators), such as in the Donation Game and the Public Goods Game, we now relax this assumption and allow the payoff difference to include the population composition. Although this generalisation introduces significant mathematical challenges, it allows for a more realistic modelling framework. Furthermore, in contrast to the aforementioned papers where the dynamics started equally-likely in the homogeneous states or in the all-defector state, we now assume a general starting state, dependant on the number of cooperators and defectors computed at equilibrium.

Herein, we present a detailed analysis of the cost functions associated with reward, punishment, and hybrid incentive schemes in the context of both General $2\times2$ Games and the Collective Risk Game. Our results can be summarised as follows.

\begin{enumerate}[(i)]
    \item We obtain analytically the asymptotic behaviour of the cost functions in the limit of neutral drift for any General $2\times 2 $ Game as well as for the Collective Risk Game.
    \item We obtain analytically the asymptotic behaviour of the cost functions in the limit of strong selection for a cooperative and defective Prisoner's Dilemma (types of General $2\times 2 $ Game) as well as for the Collective Risk Game.
    \item We analyse numerically the behaviour of the reward, punishment, and hybrid cost functions for the main classes of General $2\times2$ Games and the Collective Risk Game, highlighting the occurrence of a phase transition. 
    \end{enumerate} 

\subsection*{Organisation of the paper}The rest of the paper is organised as follows. Section \ref{sec: models} presents the model and the methods as well as the games we are interested in throughout this work. Section \ref{sec: asymptotic limits} details the asymptotic behaviour of the cost functions in neutral drift, in Lemma \ref{lem: neutral drift limit paper 3}, and in strong selection, in Lemma \ref{lem: strong selection limit paper 3}, while Section \ref{sec: numerical simulations} offers an overview of the behaviour of the incentive cost functions under various parameters. Section \ref{sec: discussion} presents a summary of our results as well as future research paths. Various calculations of fundamental matrices and limiting behaviours can be found in Appendix \ref{sec: appendix}.

\section{Model and methods} 
\label{sec: models}

This section presents the model and methods of our paper, starting with an introduction to the games of interest of this work. 

\subsection{Evolutionary dynamics} 
\label{subsec: ev dynamics}

\noindent We consider a well-mixed, finite population of $N$ self-regarding individuals who engage with one another using one of the following one-shot, i.e. non-repeated, games, a General $2\times2$ Game or the Collective Risk Game. Each player can choose either to cooperate ($C$) or to defect ($D$). 

We model the finite population dynamics on an absorbing Markov chain of $(N+1)$ states, $\{S_0, ..., S_N\}$, where a state $S_j$ represents a population with $j$ cooperators and $N-j$ defectors. The states $S_0$ and $S_N$ are absorbing. We employ Fermi's strategy update rule \cite{traulsen2006} stating that a player $X$ with fitness $f_X$ adopts the strategy of another player $Y$ with fitness $f_Y$ with a probability given by 

$$
P_{X,Y}=\left(1 + e^{-\beta(f_Y-f_X)}\right)^{-1},
$$ where $\beta$ represents the intensity of selection. 

We assume a small mutation rate, meaning that the population is monomorphic for the majority of the dynamics and that, in the rare cases that a mutation arises, it can either fixate or go extinct. This type of dynamics is most appropriately modelled by an absorbing Markov chain like the one described above which follows the timeline from when a mutation first appears (when the first defector/cooperator arises) to either extinction (state $S_N$/$S_0$) or fixation (state $S_0$/$S_N$). 

\subsection{Games}

\subsubsection*{General $2\times2$ Games} 

\noindent There are two possible strategies in a general $2\times2$ game, which, for the purpose of this paper, we call cooperate $C$ and defect $D$, with the payoff matrix for the row player given as follows  
\[
 \bordermatrix{~ & C & D\cr
                  C & R & S \cr
                  D & T & P  \cr
                 }. 
\]

\noindent We obtain the average payoff of a cooperator and a defector from the above payoff matrix, denoting by $\pi_{X,Y}$ the payoff of a player $X$ when interacting with another player $Y$:
\begin{equation*} 
\begin{split} 
\Pi_{C}(j) &=\frac{(j-1)\pi_{C,C} + (N-j)\pi_{C,D}}{N-1} = \frac{(j-1) R + (N-j) S}{N-1}  ,\\
\Pi_{D}(j) &=\frac{(N-j-1)\pi_{D,D} + j\pi_{D,C}}{N-1} =\frac{j T + (N-j-1)P}{N-1}.
\end{split}
\end{equation*}

\noindent Thus, 
\begin{align*}
    \delta_j &= \Pi_{C}(j) - \Pi_{D}(j) \\
    &= \frac{j(R-S-T+P)+N(S-P)-R+P}{N-1} \\
    &= \frac{(N-j)(S-P) + j(R-T) + (P-R)}{N-1}.
\end{align*}

\begin{remark}
    The payoff difference $\delta_j$ is independent of the number of cooperators $j$ when $R-S-T+P=0$, i.e. when $R+P=T+S$ corresponding to a Donation Game which was previously studied in the context of institutional incentives in \cite{duong2021cost, DuongDurbacHan2022, DuongDurbacHan2024}. 
\end{remark}
\noindent See \cite{fatima2024learning} for a recent survey on the social dilemmas. See below for a classification of General $2\times2$ Games:

\begin{table}[H]
\hspace*{3.5cm}
\begin{adjustbox}{width=0.45\textwidth}
\begin{tabular}{|c|l|l|}
\hline
\textbf{Type of Game}                    & \multicolumn{1}{c|}{\textbf{Game}}                      & \multicolumn{1}{c|}{\textbf{Ordering}}                    \\ \hline
\multirow{6}{*}{Dominance Games}         & \begin{tabular}[c]{@{}l@{}}PD (D)\\ PD (C)\end{tabular} & \begin{tabular}[c]{@{}l@{}}$T>R>P>S$\\ $S>P>R>T$\end{tabular} \\ \cline{2-3} 
                                         &                                                         & \begin{tabular}[c]{@{}l@{}}$R>T>S>P$\\ $P>S>T>R$\end{tabular} \\ \cline{2-3} 
                                         &                                                         & \begin{tabular}[c]{@{}l@{}}$R>S>T>P$\\ $P>T>S>R$\end{tabular} \\ \cline{2-3} 
                                         &                                                         & \begin{tabular}[c]{@{}l@{}}$R>S>P>T$\\ $P>T>R>S$\end{tabular} \\ \cline{2-3} 
                                         &                                                         & \begin{tabular}[c]{@{}l@{}}$S>R>P>T$\\ $T>P>R>S$\end{tabular} \\ \cline{2-3} 
                                         &                                                         & \begin{tabular}[c]{@{}l@{}}$S>R>T>P$\\ $T>P>S>R$\end{tabular} \\ \hline
\multirow{3}{*}{Anti-Coordination Games} & \begin{tabular}[c]{@{}l@{}}HD (D)\\ HD (H)\end{tabular} & \begin{tabular}[c]{@{}l@{}}$T>R>S>P$\\ $S>P>T>R$\end{tabular} \\ \cline{2-3} 
                                         &                                                         & \begin{tabular}[c]{@{}l@{}}$S>T>R>P$\\ $T>S>P>R$\end{tabular} \\ \cline{2-3} 
                                         &                                                         & \begin{tabular}[c]{@{}l@{}}$S>T>P>R$\\ $T>S>R>P$\end{tabular} \\ \hline
\multirow{3}{*}{Coordination Games}      & \begin{tabular}[c]{@{}l@{}}SH (S)\\ SH (H)\end{tabular} & \begin{tabular}[c]{@{}l@{}}$R>T>P>S$\\ $P>S>R>T$\end{tabular} \\ \cline{2-3} 
                                         &                                                         & \begin{tabular}[c]{@{}l@{}}$R>P>S>T$\\ $P>R>T>S$\end{tabular} \\ \cline{2-3} 
                                         &                                                         & \begin{tabular}[c]{@{}l@{}}$R>P>T>S$\\ $P>R>S>T$\end{tabular} \\ \hline
\end{tabular}
  \end{adjustbox}
\caption{Classification of General $2\times2$ Games  by the ordering of the entries of the payoff matrix, with known games labelled: Prisoner's Dilemma (PD), Hawk and Dove or Snowdrift (HD), and Stag Hunt (SH). There are a total of 24 orderings, corresponding to 12 independent $2\times2$ games. The games are split into three types – dominance, anti-coordination, and coordination – depending on the equilibria observed in the asymptotic limit $N\rightarrow \infty$. See \cite{pires2022more} for details.}
\label{tab: general 2x2 games classification}
\end{table}

\begin{remark}
For the Defective Prisoner's Dilemma, where $T>R>P>S$, the payoff difference $\delta_j<0$ for all $j$, while for the Cooperative Prisoner's Dilemma, where $S>P>R>T$, $\delta_j>0$ for all $j$. For all other General $2\times2$ Games, the sign of $\delta_j$ is not-trivial. In some of our analytical results, we focus our attention on the defective and cooperative Prisoner's Dilemma for tractability reasons.
\end{remark}

\subsubsection*{Collective Risk Game} 

\noindent The Collective Risk Game (CRG) is a cooperation dilemma in which a population of $N$ players engage in an $n$-person dilemma where each individual is able to contribute or not to a common good, that is, to cooperate ($C$) or to defect ($D$), respectively. Cooperation corresponds to offering a fraction $c$ of their endowment $B$, while defection means offering nothing. If a group of size $n$ does not contain $m$ cooperators (a collective effort of $mcB$), then the members will lose their remaining endowments with a probability $r$, otherwise they keep what they have. Intuitively, defecting brings a larger payoff to the individual player, but could lead to collective failure due to the inability to satisfy the number of cooperators needed in order to keep the group's endowments. CRG is widely used to model the climate change dilemma.

Letting $k$ be the numbers of cooperators in a group of size $n$ and denoting by $\pi_D$ ($\pi_C$) the payoff of a defector (cooperator) in a single round, we have:
\begin{equation*} 
\begin{split} 
\pi_D(k) &= B(\eta(k-m)+(1-r)(1-\eta(k-m)))\\
\pi_C(k) &= \pi_D(k) - cB,
\end{split}
\end{equation*} 
where $\eta(x<0)=0$ and $\eta(x\geq 0)=1$ (Heaviside step function).

In each round of the game, a group of 
$n$ players is randomly selected from the population of $N$ players through a process of sampling without replacement. The likelihood of selecting any specific combination of cooperators and defectors in a group follows a hypergeometric distribution. Within a selected group, each individual's strategy is associated with a payoff, representing their earnings for that round. Fitness is defined as the expected payoff for an individual across the entire population. In a population with $j$ cooperators out of $N$ individuals, where a selected group of size $n$ contains $k$ cooperators, the fitness of a defector ($D$) and a cooperator ($C$) can be expressed as follows:

\begin{equation*} 
\begin{split} 
\Pi_D(j) &= \frac{1}{{N-1 \choose n-1}}\sum_{k=0}^{n-1}{j \choose k}{N-j-1 \choose n-k-1}\pi_D(k)\\
\Pi_C(j) &= \frac{1}{{N-1 \choose n-1}}\sum_{k=0}^{n-1}{j-1 \choose k}{N-j \choose n-k-1}\pi_C(k + 1).
\end{split}
\end{equation*} 


\noindent Hence, 
\begin{align*}
    \Pi_C(j) - \Pi_D(j) &= \frac{1}{{N-1 \choose n-1}}\sum_{k=0}^{n-1}{j-1 \choose k}{N-j \choose n-k-1}\pi_C(k + 1) - \frac{1}{{N-1 \choose n-1}}\sum_{k=0}^{n-1}{j \choose k}{N-j-1 \choose n-k-1}\pi_D(k). \\
\end{align*}
We denote $\Pi_C(j) - \Pi_D(j)\coloneqq \delta_j$.\\

\begin{remark}
    In the case of both General $2\times 2$ Games and the Collective Risk Game, the difference in payoffs $\delta_j$ depends on $j$, the state of the Markov chain. This makes the computation of the inverse of the transition matrix intractable (needed for obtaining close-form cost functions), however in the neutral drift and strong selection limits we can obtain analytical limiting results. For the general behaviour of the cost functions, we employ numerical simulations.
\end{remark}



\subsection{Deriving the incentive cost function}

\label{sec: cost of incentive}

\subsubsection*{Institutional incentives} 

We assume the existence of an external decision-maker (called an institution) who has the budget to intervene in a population in order to steer it towards a desire outcome, i.e. to prosocial behaviours. Following the approach in \cite{duong2021cost, DuongDurbacHan2022, DuongDurbacHan2024}, we employ the following incentive scheme: to reward a cooperator (punish a defector), the institution must expend an amount of $\theta/a$  ($\theta/b$, respectively), ensuring that the cooperator's (defector's) payoff increases (decreases) by $\theta$, where $a,b>0$ are constants that define the efficiency of delivering each type of incentive. The cost per capita $\theta$ is a fixed parameter, remaining unchanged throughout the dynamics. Within this institutional enforcement framework, we assume that the institution possesses complete information regarding the population's composition or statistical attributes at the time of decision-making. Specifically, given the well-mixed population model, we assume that the number $j$ of cooperators in the population is known.

Hence, we have that, for $1\leq j\leq N-1$, the  cost per generation $j$ for the incentive providing  institution is
\begin{equation}
\label{eq: incentives per generation}
\theta_j = \begin{cases} \frac{j}{a}\theta\quad \text{reward incentive},\\
\frac{N-j}{b}\theta\quad \text{punishment incentive},\\
\min\Big(\frac{j}{a} \frac{N-j}{b}\Big)\theta\quad\text{mixed incentive}.
\end{cases}
\end{equation}

\subsubsection*{Cooperation frequency}

Recall the evolutionary dynamics modelled on an absorbing  Markov chain as presented in Subsection \ref{subsec: ev dynamics}. Using this information, we compute the expected number of times the population contains $j$ C players for $1 \leq j \leq N-1$. Let $U = \{u_{ij}\}_{i,j = 1}^{N-1}$ denote the transition matrix between the $N-1$ transient states of our Markov chain, $\{S_1, ..., S_{N-1}\}$. The transition probabilities can be defined as follows, for $1\leq i \leq N-1$: 
\begin{equation} 
\label{eq: transition probabilities}
\begin{split} 
u_{i,i\pm k} &= 0 \qquad \text{ for all } k \geq 2, \\
u_{i,i\pm1} &= \frac{N-i}{N} \frac{i}{N} \left(1 + e^{\mp\beta[\Pi_C(i) - \Pi_D(i)+\theta]}\right)^{-1},\\
u_{i,i} &= 1 - u_{i,i+1} -u_{i,i-1},
\end{split}
\end{equation} where $\Pi_C(i)$ and $\Pi_D(i)$ are the average payoffs of a cooperator (C) and defector (D), respectively, at state $i$.

Since the population consists of only two strategies, the fixation  probabilities of a C (D) player in a homogeneous population of D (C) players  when the interference scheme is carried out are, respectively, \cite{nowak}
\begin{equation*} 
\begin{split}
\rho_{D,C} &= \left(1+\sum_{i = 1}^{N-1} \prod_{k = 1}^i \frac{1+e^{\beta(\Pi_C(k)-\Pi_D(k) + \theta)}}{1+e^{-\beta(\Pi_C(k)-\Pi_D(k)+\theta)}}  \right)^{-1}, \\
\rho_{C,D} &= \left(1+\sum_{i = 1}^{N-1} \prod_{k = 1}^i \frac{1+e^{\beta(\Pi_D(k)-\Pi_C(k) - \theta)}}{1+e^{-\beta(\Pi_D(k)-\Pi_C(k)-\theta)}}  \right)^{-1}.
\end{split}
\end{equation*} 

Computing the stationary distribution using  these fixation probabilities, we  obtain the frequency of cooperation  $$\frac{\rho_{D,C}}{\rho_{D,C}+\rho_{C,D}}.$$

Hence, this frequency of cooperation can be maximised by maximising 
\begin{equation}
\label{eq:max}
\max_{\theta} \left(\rho_{D,C}/\rho_{C,D}\right).  
\end{equation} 

The fraction in Equation~\eqref{eq:max} can be simplified as follows \cite{nowak2006} 
\begin{eqnarray}
\nonumber
\frac{\rho_{D,C}}{\rho_{C,D}} &=&  \prod_{k = 1}^{N-1} \frac{u_{i,i-1}}{u_{i,i+1}} =\prod_{k = 1}^{N-1} \frac{1 + e^{\beta[\Pi_C(k)-\Pi_D(k) + \theta]}}{1 + e^{-\beta[\Pi_C(k)-\Pi_D(k) + \theta]}} \\
\nonumber
&=& e^{\beta\sum_{k = 1}^{N-1} \left(\Pi_C(k)-\Pi_D(k) + \theta\right)} \\
\label{eq:max_Q_prime}
 &=& e^{\beta [(N-1)\theta +  \sum_{k=1}^{N-1}\delta_k]},\nonumber\\
 &=& e^{\beta [(N-1)(\theta +  \Delta)]},
 \end{eqnarray} 
 where $u_{i,i-1}$ and $u_{i,i-1}$ are the probabilities  to decrease or increase the number  of $C$ players  (i.e. $i$) by one in each time step, respectively, and $\Delta=
\frac{1}{N-1}\sum_{k=1}^{N-1}\delta_k$ is the average of the payoff differences between a cooperator and a defector.

We consider non-neutral selection, i.e.  $\beta > 0$ (under neutral selection, there is no need to use incentives as the likelihood of a player imitating another is low and any changes in strategy are due to noise as opposed to payoffs). Assuming that we desire to obtain  at least an $\omega  \in [0,1]$ fraction of cooperation, i.e. $\frac{\rho_{D,C}}{\rho_{D,C}+\rho_{C,D}} \geq \omega$, it follows from equation~\eqref{eq:max_Q_prime}  that
\begin{equation} 
\label{eq:omega_fraction}
 \theta \geq \theta_0(\omega) = \frac{1}{(N-1)\beta} \log\left(\frac{\omega}{1-\omega}\right) - \Delta.
\end{equation}
Therefore it is guaranteed that if $\theta  \geq \theta_0(\omega)$, at least an $\omega$ fraction of cooperation can be expected.
This condition implies that the lower bound of $\theta$ monotonically depends on $\beta$. Namely, when $\omega \geq 0.5$, it increases with $\beta$ and when $\omega < 0.5$, it decreases with $\beta$.

\subsubsection*{Cost function}

We now derive the total expected cost of inference over all generations. Let $(n_{ik})_{i,k=1}^{N-1}$ be the entries of the fundamental matrix of the absorbing Markov chain of the evolutionary process. These aforementioned entries give the expected number of times the population is in the state $S_j$ if it has started in the transient state $S_i$ \cite{kemeny1976finite}. The population will start, in the long run, at $i = 0$ ($i = N$, respectively) with probability equal to the frequency of defectors $D$ (cooperators $C$) computed at the equilibrium, $f_D = 1/(r+1)$ ($f_C = r/(r+1)$, respectively), where $r = e^{\beta (N-1)(\Delta +  \theta)}$. Thus, generally, the expected number of visits at state $S_i$ will be $ f_D n_{1i} + f_C n_{N-1,i}$. Therefore, the expected cost of inference over all generations is given by
\begin{equation}
\label{eq: incentives general}
E(\theta) = \sum_{j=1}^{N-1} (f_D n_{1,j} + f_C n_{N-1,j} )\theta_j.
\end{equation}
To summarise, we obtain the following constrained minimisation problem
\begin{equation}
\label{eq: min prob} \min_{\theta\geq \theta_0} E(\theta),
\end{equation} where $E(\theta)$ can be either $E_r(\theta)$, $E_p(\theta)$, or $E_{mix}(\theta)$.

\section{Neutral drift and strong selection limits}
\label{sec: asymptotic limits}

In this section, we compute the fundamental matrix of the Markov chain of the evolutionary process for reward incentives, noting that we obtain similar results for punishment and hybrid incentives (see Appendix \ref{sec: fundamental matrices punishment/hybrid}). 

\noindent For reward, we have:
\begin{equation} 
\label{eq: transition probabilities reward General 2x2 Games}
\begin{split} 
u_{i,i\pm k} &= 0 \qquad \text{ for all } k \geq 2, \\
u_{i,i\pm1} &= \frac{N-i}{N} \frac{i}{N} \left(1 + e^{\mp\beta[\delta_i +\theta_i/i]}\right)^{-1},\\
u_{i,i} &= 1 - u_{i,i+1} -u_{i,i-1}.
\end{split} 
\end{equation}

We normalise $a=1$ for simplicity and obtain (recalling that $\Pi_C(i)-\Pi_D (i)=\delta_i$ and that $\theta_i/i=\frac{\theta}{a}=\theta$ for $1\leq i\leq N-1$):
\begin{equation} 
\label{eq: transition probabilities reward General 2x2 Games normalised}
\begin{split} 
u_{i,i\pm k} &= 0 \qquad \text{ for all } k \geq 2, \\
u_{i,i\pm1} &= \frac{N-i}{N} \frac{i}{N} \left(1 + e^{\mp\beta[\delta_i +\theta]}\right)^{-1},\\
u_{i,i} &= 1 - u_{i,i+1} -u_{i,i-1}.
\end{split} 
\end{equation}
Next, we need to calculate the entries $n_{ik}$ of the fundamental matrix $\mathcal{N}=(n_{ik})_{i,k=1}^{N-1}= (I-U)^{-1}$. By using $\frac{1}{1+m}+\frac{1}{1+\frac{1}{m}}=1$ for $m=e^{\mp\beta[\delta_i +\theta]}$, we get $u_{i,i+1}+u_{i,i-1}=\frac{N-i}{N}\frac{i}{N}$. Then, by letting $V=(I-U)$, we obtain: 
\begin{equation} 
\begin{split} 
v_{i,i\pm k} &= 0 \qquad \text{ for all } k \geq 2, \\
v_{i,i\pm1} &= -\frac{N-i}{N} \frac{i}{N} \left(1 + e^{\mp\beta[\delta_i +\theta]}\right)^{-1},\\
v_{i,i} &= u_{i,i+1} +u_{i,i-1}=\frac{N-i}{N}\frac{i}{N}.
\end{split} 
\end{equation}

We can further write $V=W \mathrm{diag}\Big\{\frac{N-1}{N}\frac{1}{N},\ldots,\frac{N-i}{N}\frac{i}{N},\ldots, \frac{1}{N}\frac{N-1}{N}\Big\}$, where 
\begin{equation}
\label{eq: matrix W General 2x2 Games}
W=\begin{pmatrix}
1&-a_1&&&&&\\
-c_1&1&-a_2&&&&&\\
&\ddots&\ddots&\ddots&&&\\
&&-c_i&1&-a_i&&&\\
&&&-c_{i+1}&1&-a_{i+1}&&&\\
&&&&\ddots&\ddots&\ddots&\\
&&&&&-c_{N-2}&1&-a_{N-1}\\
&&&&&&-c_{N-1}&1
\end{pmatrix},
\end{equation}
with $a_i:=(1+e^{-\beta(\delta_i+\theta)})^{-1}, c_i:=(1+e^{\beta(\delta_i+\theta)})^{-1}$. 

This implies  that
$\mathcal{N}=V^{-1}=\mathrm{diag}\Big\{\frac{N^2}{N-1},\frac{N^2}{2(N-2)},\ldots,\frac{N^2}{N-1}\Big\}W^{-1}$,
and so, the fundamental matrix is $\mathcal{N}=(n_{ik})_{i,k=1}^{N-1}=\mathrm{diag}\Big\{\frac{N^2}{N-1},\frac{N^2}{2(N-2)},\ldots,\frac{N^2}{N-1}\Big\} W^{-1}$. \\

The incentive cost function \eqref{eq: incentives general} can be rewritten using the entries of $W$ as follows:
\begin{align}
\label{eq: cost function with fC and fD}
  E(\theta)&=N^2\sum_{j=1}^{N-1} \frac{\theta_j}{j(N-j)} (f_D W^{-1}_{1,j}+ f_C W^{-1}_{N-1,j})
  \\&=\begin{cases}
  N^2\theta\sum\limits_{j=1}^{N-1} \frac{1}{N-j}(f_D W^{-1}_{1,j}+ f_C W^{-1}_{N-1,j})\quad \text{for reward incentives},\\
  N^2\theta\sum\limits_{j=1}^{N-1} \frac{1}{j}(f_D W^{-1}_{1,j}+ f_C W^{-1}_{N-1,j})\quad \text{for punishment incentives},\\
  N^2\theta\sum\limits_{j=1}^{N-1} \frac{\min(\frac{j}{a},\frac{N-j}{b})}{j(N-j)}(f_D W^{-1}_{1,j}+ f_C W^{-1}_{N-1,j})\quad \text{for hybrid incentives}. \nonumber
  \end{cases}  
\end{align}

Inverting the fundamental matrix of the Markov chain $W$ to obtain the cost function $E(\theta)$ in closed form is challenging. We can, however, get its limiting behaviour in neutral drift ($\beta\rightarrow 0$) and in strong selection ($\beta\rightarrow\infty$) for a specific $2\times2$ game, namely the cooperative Prisoner's Dilemma where $\delta_i+\theta>0$ for all $i$ (as $\delta_i>0$). In the defective Prisoner's Dilemma, $\delta_i<0$ for all $i$. In this situation, when computing the limit of strong selection, we will assume that $\theta$ is small enough so that $\delta_i+\theta<0$ or that $\theta$ is large enough so that $\delta_i+\theta>0$ for all $i$. In the Collective Risk Game, the sign of $\delta_i$ changes with $i$. Therefore, we assume that $\theta$ is large enough so that $\delta_i+\theta>0$ for all $i$.

Before we present the limiting behaviour, we require the following technical lemma:

\begin{lemma}(\cite{Huang1997})
\label{lem: Toeplitz inverse}
Let $T$ be a general Toeplitz tridiagonal matrix of the form
$$
T=\left(\begin{array}{ccccccc}
1 & -u & & & & &  \\
-l & 1 & -u & & & & \\
& \ddots & \ddots & \ddots & & & \\
& & -l & 1 & -u & & \\
& & & \ddots & \ddots & \ddots & \\
& & & & -l & 1 & -u \\
& & & & & -l & 1
\end{array}\right).
$$ 
Then the inverse of $T$ is given by
$$
\left(T^{-1}\right)_{i j}=\frac{\left(\lambda_{+}^{i}-\lambda_{-}^{i}\right)\left(\lambda_{+}^{n-j+1}-\lambda_{-}^{n-j+1}\right)}{\left(\lambda_{+}-\lambda_{-}\right)\left(\lambda_{+}^{n+1}-\lambda_{-}^{n+1}\right)} u^{j-i} \hspace{0.5cm} i<j
$$
$$
\left(T^{-1}\right)_{i j}=\frac{\left(\lambda_{+}^{j}-\lambda_{-}^{j}\right)\left(\lambda_{+}^{n-i+1}-\lambda_{-}^{n-i+1}\right)}{\left(\lambda_{+}-\lambda_{-}\right)\left(\lambda_{+}^{n+1}-\lambda_{-}^{n+1}\right)} l^{i-j} \hspace{0.5cm} i \geq j,
$$
where
\begin{equation*}
\lambda_{+}=\frac{1+\sqrt{1-4 l u}}{2}, \quad \lambda_{-}=\frac{1-\sqrt{1-4 l u}}{2}.
\end{equation*}
\end{lemma}

\begin{lemma}(Neutral drift limit)
\label{lem: neutral drift limit paper 3}
For both a general $2\times 2$ game and the Collective Risk Game, we have, for reward incentives,
$$\lim_{\beta\to 0}E_r(\theta)= N^2\theta H_N,$$
where $H_N$ is the harmonic number
\[
H_N=\sum_{i=1}^{N-1}\frac{1}{j}.
\]
\end{lemma}
\begin{proof}
\noindent It follows from the formula of the entries of the fundamental matrix $W$ given in \eqref{eq: matrix W General 2x2 Games} that, in the limit of neutral drift, i.e. when $\beta\rightarrow0$, we have $\lim_{\beta\rightarrow 0} W=\bar W$, where
\begin{align}
\bar W= \frac{1}{2}\begin{pmatrix} 
2&-1&&&&&\\
-1&2&-1&&&&&\\
&\ddots&\ddots&\ddots&&&\\
&&-1&2&-1&&&\\
&&&-1&2&-1&&&\\
&&&&\ddots&\ddots&\ddots&\\
&&&&&-1&2&-1\\
&&&&&&-1&2
\end{pmatrix}.
\end{align}
The tridiagonal matrix above is the well-known Cartan matrix. It follows from \cite{wei2017inverses} that $(\bar{W}^{-1})_{i,j}= 2\Big[\min(i,j)-\frac{ij}{N}\Big]$, for $1\leq i,j\leq N-1.$

Then, noting that  $\lim\limits_{\beta\to 0} f_C = \lim\limits_{\beta\to 0} f_D = \frac{1}{2}$, we have
\begin{equation}
\label{eq: neutral drift limit}
    \lim_{\beta\to 0}E(\theta)= \lim_{\beta\to 0}\sum_{j=1}^{N-1} (f_Dn_{1,j} + f_Cn_{N-1,j} )\theta_j = \frac{N^2}{2}\sum_{j=1}^{N-1} \frac{(\bar W^{-1}_{1,j} + \bar W^{-1}_{N-1,j})}{j(N-j)}\theta_j.
\end{equation}

Assuming without loss of generality that $a = 1$ and substituting $\theta_j=j\theta$ in the case of reward, we obtain
\begin{align*}
    \lim_{\beta\to 0}E_r(\theta) &=\frac{N^2}{2}\sum_{j=1}^{N-1} \frac{(\bar W^{-1}_{1,j} + \bar W^{-1}_{N-1,j})}{j(N-j)} \theta_j \\
    &= \frac{N^2}{2}\sum_{j=1}^{N-1} \frac{(\bar W^{-1}_{1,j} + \bar W^{-1}_{N-1,j})}{N-j}\theta \\
    &= \frac{N^2}{2}\Big[\sum_{j=1}^{N-1} \frac{\bar W^{-1}_{1,j}}{N-j} \theta + \sum_{j=1}^{N-1} \frac{\bar W^{-1}_{N-1,j}}{N-j} \theta\Big].
\end{align*} 
Using the fact that $\bar W^{-1}_{1,j}=\frac{2(N-j)}{N}$ and $\bar W^{-1}_{N-1,j}=\frac{2j}{N}$, we have
\begin{align*}
    \lim_{\beta\to 0}E_r(\theta) &= \frac{N^2}{2}\Big[\sum_{j=1}^{N-1} \frac{\bar W^{-1}_{1,j}}{N-j} \theta + \sum_{j=1}^{N-1} \frac{\bar W^{-1}_{N-1,j}}{N-j} \theta\Big] \\
    &= \frac{N^2}{2}\Big[\sum_{j=1}^{N-1} \frac{2(N-j)}{N(N-j)} \theta + \sum_{j=1}^{N-1} \frac{2j}{N(N-j)} \theta\Big] \\
    &= \frac{N^2\theta}{2}\frac{2}{N}\Big[ \sum_{j=1}^{N-1} 1 + \sum_{j=1}^{N-1} \frac{j}{N-j} \Big] \\
    &= N \theta \Big[ (N-1)+ \sum_{k=1}^{N-1} \frac{N-k}{k} \Big] \\
    &=N \theta\Big[ (N-1) + \sum_{k=1}^{N-1} (\frac{N}{k} - 1) \Big] \\
    &= N\theta \Big[ (N-1) + \sum_{k=1}^{N-1} \frac{N}{k} - \sum_{k=1}^{N-1} 1 \Big] \\
    &= N \theta\Big[ (N-1) + NH_{N}-(N-1) \Big] \\
    &= N^2 \theta H_{N}.
\end{align*}
The limiting behaviour for punishment and hybrid incentives can be found in Appendix \ref{sec: appendix neutral drift}.
\end{proof}

\begin{remark}
We notice that the neutral limit of the cost function does not depend on the specific game. This has been shown previously in our papers~ \cite{duong2021cost,DuongDurbacHan2024} for the Donation Game and the Public Goods Game. This is expected since, in the neutral drift, the underlying game does not influence the evolutionary dynamics.

Note further that, in this paper, the method we employed to compute the neutral drift limit is different from that in \cite{duong2021cost,DuongDurbacHan2024}. More precisely, in \cite{duong2021cost,DuongDurbacHan2024},
we were able to compute explicitly the cost functions and then derive the neutral drift limit. In this work, we cannot compute explicitly the cost functions since inverting the matrix $W$ (needed for a closed-form of these) is intractable. Therefore, we obtain the neutral drift limit of the cost functions by calculating this limit for the transition matrix first.
\end{remark}

\begin{figure}[H]
\centering
\hspace*{-0.5cm}  
\includegraphics[width=\textwidth]{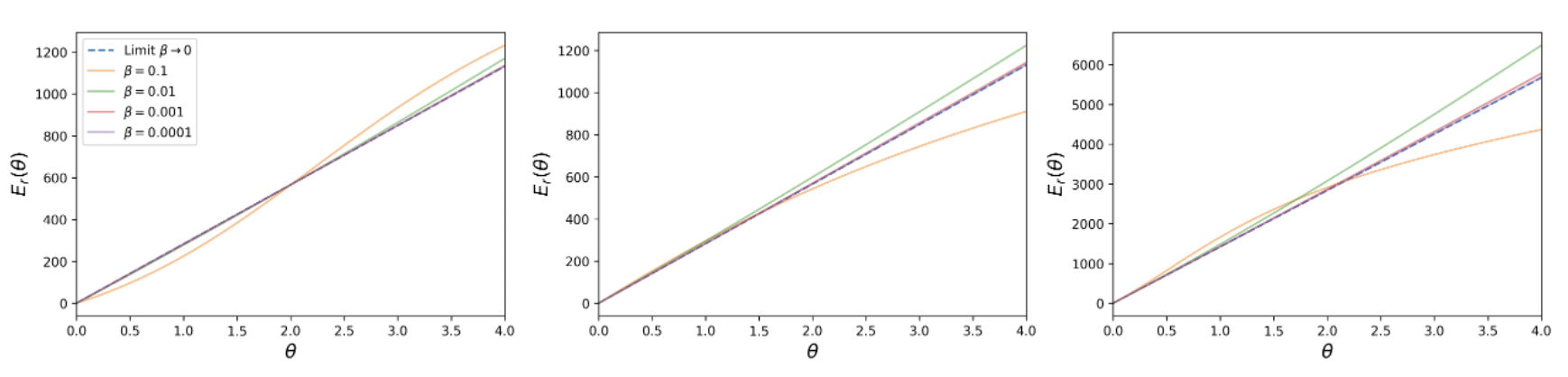}
\caption{The dashed blue line represents the neutral drift limit of the reward cost function. The three figures illustrate this behaviour for: (i) a defective Prisoner's Dilemma (left; $R = 2$, $S = 0$, $T = 4$, $P = 1$), (ii) a cooperative Prisoner's Dilemma (centre; $R = 1$, $S = 4$, $T = 0$, $P = 2$), and (iii) the Collective Risk Game (right; $B = 1$, $c = 0.1$, $r = 0.9$, $n = 5$, $m = 2$). As the intensity of selection $\beta$ decreases, the reward cost function converges towards the neutral drift limit, in accordance to our theoretical results in Lemma \ref{lem: neutral drift limit paper 3}.} 
\label{fig: neutral drift limit PD}
\end{figure}

\begin{lemma}(Strong selection limit)
\label{lem: strong selection limit paper 3}
\begin{enumerate}
    \item \label{lem: ssl part 1} Defective Prisoner's Dilemma - Recall that, for the defective Prisoner's Dilemma, $\delta_j < 0$ for all $j$. Assume that $\theta$ is sufficiently small so that $\delta_j + \theta < 0$ for all $j$. Then, in the limit of strong selection, under reward incentives, 
\[
\lim_{\beta \to \infty} E_r(\theta) = \frac{N^2\theta}{N-1}.
\]
\item \label{lem: ssl part 2} Cooperative and defective Prisoner's Dilemma, the Collective Risk Game - Recall that, for the cooperative Prisoner's Dilemma, $\delta_j > 0$ for all $j$. Assume that $\theta$ is sufficiently large so that $\delta_j + \theta > 0$ for all $j$ also holds in the defective Prisoner's Dilemma and in the Collective Risk Game. Then, in the limit of strong selection, under reward incentives, 
\[
\lim_{\beta \to \infty} E_r(\theta) = N^2 \theta
\]
\end{enumerate}
\end{lemma}
\begin{proof}

The fundamental matrix $W$ transforms as following in the limit of strong selection, i.e. when $\beta\rightarrow\infty$, in part \ref{lem: ssl part 1} and part \ref{lem: ssl part 2} of the lemma:
\begin{equation}
\label{eq: matrix W in beta infy limit for delta <= 0}
\hat W=\begin{pmatrix}
1&0&&&&&\\
-1&1&0&&&&&\\
&\ddots&\ddots&\ddots&&&\\
&&-1&1&0&&&\\
&&&-1&1&0&&&\\
&&&&\ddots&\ddots&\ddots&\\
&&&&&-1&1&0\\
&&&&&&-1&1
\end{pmatrix},
\end{equation}
and
\begin{equation}
\label{eq: matrix W in beta infy limit for delta > 0}
\overline{W} =\begin{pmatrix} 
1&-1&&&&&\\
0&1&-1&&&&&\\
&\ddots&\ddots&\ddots&&&\\
&&0&1&-1&&&\\
&&&0&1&-1&&&\\
&&&&\ddots&\ddots&\ddots&\\
&&&&&0&1&-1\\
&&&&&&0&1
\end{pmatrix}.
\end{equation}

Using Lemma \ref{lem: Toeplitz inverse}, we can compute the matrix inverse for both \eqref{eq: matrix W in beta infy limit for delta <= 0} and \eqref{eq: matrix W in beta infy limit for delta > 0} as follows.\\

For \eqref{eq: matrix W in beta infy limit for delta <= 0}, $\lambda_+=1$ and $\lambda_-=0$, thus
\begin{equation*}
\hat W^{-1}_{i,j} = 
\begin{cases}
     0 \hspace{0.5cm} i<j \quad \\
     1 \hspace{0.5cm} i\geq j 
\end{cases}    
\end{equation*}

Therefore, $\hat W^{-1} =\begin{pmatrix}
1      & 0      & 0      & \cdots & 0 \\
1      & 1      & 0      & \cdots & 0 \\
1      & 1      & 1      & \cdots & 0 \\
\vdots & \vdots & \vdots & \ddots & \vdots \\
1      & 1      & 1      & \cdots & 1 \\
\end{pmatrix}.$\\

For \eqref{eq: matrix W in beta infy limit for delta > 0}, $\lambda_+=1$ and $\lambda_-=0$, thus
\begin{equation*}
\overline{W}^{-1}_{i,j} = 
\begin{cases}
     1 \hspace{0.5cm} i\leq j \quad \\
     0 \hspace{0.5cm} i> j 
\end{cases}    
\end{equation*}

Therefore, $\overline{W}^{-1} =\begin{pmatrix}
1 & 1 & 1 & \cdots & 1 \\
0 & 1 & 1 & \cdots & 1 \\
0 & 0 & 1 & \cdots & 1 \\
\vdots & \vdots & \vdots & \ddots & \vdots \\
0 & 0 & 0 & \cdots & 1 \\
\end{pmatrix}.$\\

In the limit of strong selection, $f_D = \frac{1}{e^{\beta (N-1)(\Delta + \theta)}+1}$ and $f_C = \frac{e^{\beta (N-1)(\Delta +  \theta)}}{e^{\beta (N-1)(\Delta +  \theta)}+1}$ transform as follows.

Recall $\Delta=
\frac{1}{N-1}\sum_{j=1}^{N-1}\delta_j$. Then
$$
\Delta+\theta= \frac{1}{N-1}\sum_{j=1}^{N-1}\delta_j + \theta = \frac{1}{N-1}\sum_{j=1}^{N-1}\Big(\delta_j+\theta\Big).
$$

In part \ref{lem: ssl part 1}, as $\delta_j+\theta<0$, we have $\Delta+\theta<0$. Similarly, in part \ref{lem: ssl part 2}, since $\delta_j+\theta>0$, we get that $\Delta+\theta>0$. Thus,

\[
\lim_{\beta \to \infty} f_C =
\begin{cases}
0 & \text{for part \ref{lem: ssl part 1}}\\
1 & \text{for part \ref{lem: ssl part 2}}
\end{cases}
\quad
\lim_{\beta \to \infty} f_D =
\begin{cases}
1 & \text{for part \ref{lem: ssl part 1}} \\
0 & \text{for part \ref{lem: ssl part 2}}.
\end{cases}
\]

Therefore, by rewriting the cost function $E(\theta)$ from Equation \eqref{eq: cost function with fC and fD}, we get:

\[
\lim_{\beta \to \infty} E(\theta) =
\begin{cases}
N^2 \sum_{j=1}^{N-1} \frac{\overline{W}_{N-1,j}}{j(N-j)}\theta_j & \text{for } \text{part \ref{lem: ssl part 1}} \\ 
N^2 \sum_{j=1}^{N-1} \frac{\hat W_{1,j}}{j(N-j)}\theta_j  & \text{for } \text{part \ref{lem: ssl part 2}}.
\end{cases}
\]

Substituting the values of $\hat W_{i,j}$ and $\overline{W} _{i,j}$ computed in \eqref{eq: matrix W in beta infy limit for delta <= 0} and \eqref{eq: matrix W in beta infy limit for delta > 0}, respectively, as well as the definition of $\theta_j$ from Equation \eqref{eq: incentives per generation}, we get

\[
\lim_{\beta \to \infty} E(\theta) =
\begin{cases}
\frac{N^2\theta}{N-1} & \text{for } \text{part \ref{lem: ssl part 1}} \\ 
N^2 \theta & \text{for } \text{part \ref{lem: ssl part 2}}.
\end{cases}
\]
\end{proof}

\begin{figure}[H]
\centering
\hspace*{-0.5cm}  
\includegraphics[width=\textwidth]{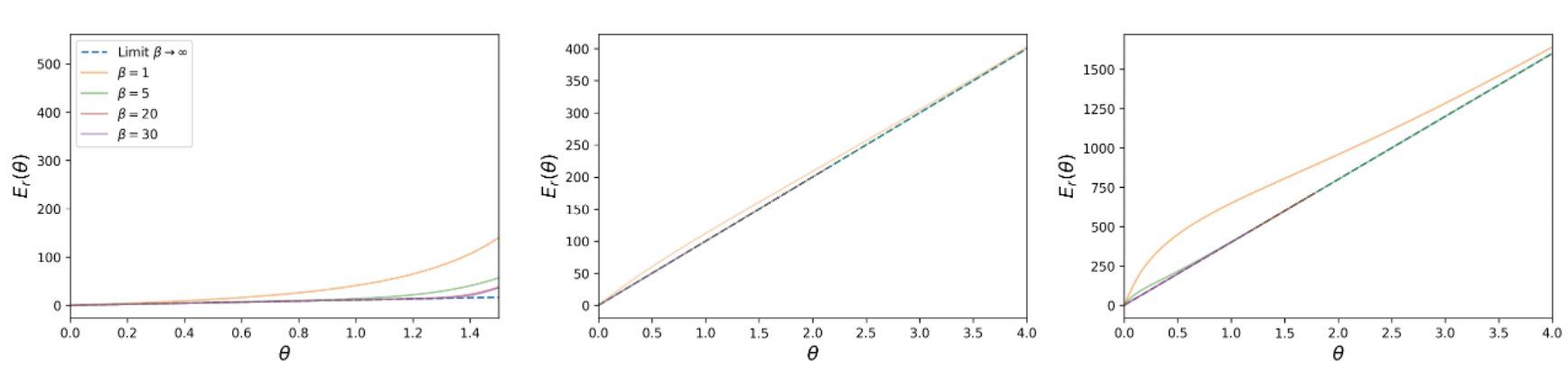}
\caption{The dashed blue line represents the strong selection limit of the reward cost function. The three figures correspond to: (i) a defective Prisoner's Dilemma (left; $R = 2$, $S = 0$, $T = 4$, $P = 1$; restricted such that $\delta + \theta < 0$), (ii) a cooperative Prisoner's Dilemma (centre; $R = 1$, $S = 4$, $T = 0$, $P = 2$), and (iii) the Collective Risk Game (right; $B = 1$, $c = 0.1$, $r = 0.9$, $n = 5$, $m = 2$; restricted such that $\delta + \theta > 0$). As the intensity of selection $\beta$ increases, the reward cost function converges towards the strong selection limit, in accordance to our theoretical results in Lemma \ref{lem: neutral drift limit paper 3}.} 
\label{fig: strong selection limit PD}
\end{figure}

\section{Numerical simulations}
\label{sec: numerical simulations}

This section contains the results of our numerical analysis, coded in Python 3.10. We employ a logarithmic scale to respond to skewness of larger values and, in the case of the Collective Risk Game, we restrict the range of $\theta$ to highlight the phase transition.

\subsection{General $2\times2$ Games}

\begin{figure}[H]
\centering
\hspace*{-0.5cm} 
\includegraphics[width=\textwidth]{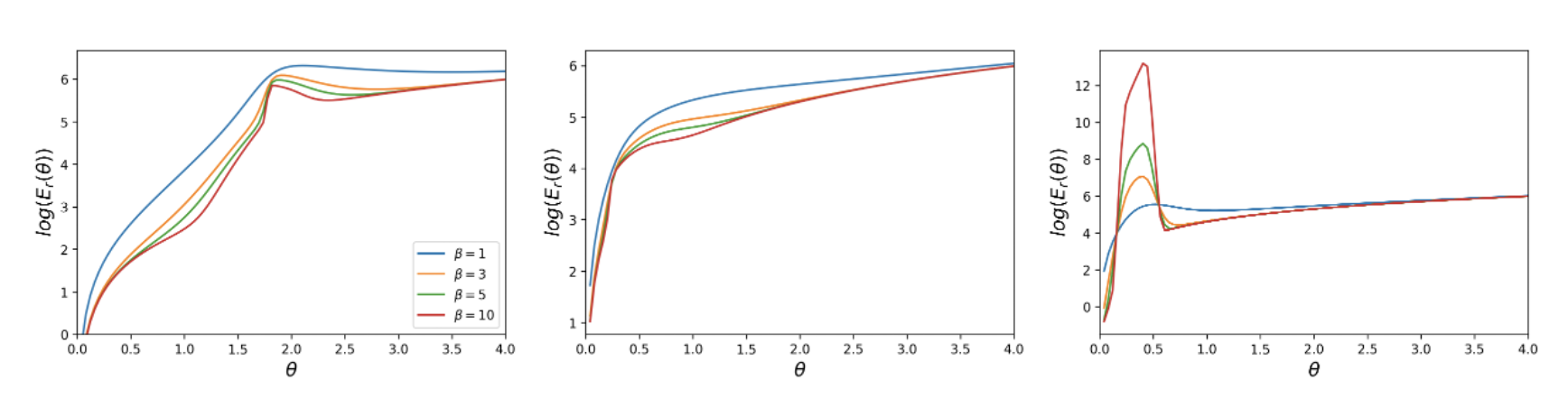}
\caption{Comparison of the reward cost function $E_r(\theta)$ for the General $2 \times 2$ Games: Prisoner's Dilemma ($R = 2$, $S = 0$, $T = 4$, $P = 1$), Hawk and Dove ($R = 2$, $S = 1$, $T = 3$, $P = 0$), and Stag Hunt ($R = 3$, $S = 0$, $T = 2$, $P = 1$), shown in this order. Each figure illustrates the behaviour of $E_r(\theta)$ for various values of the intensity of selection $\beta$, with population size $N = 10$.} 
\label{fig: reward comparison}
\end{figure}

\begin{figure}[H]
\centering
\hspace*{-0.5cm} 
\includegraphics[width=\textwidth]{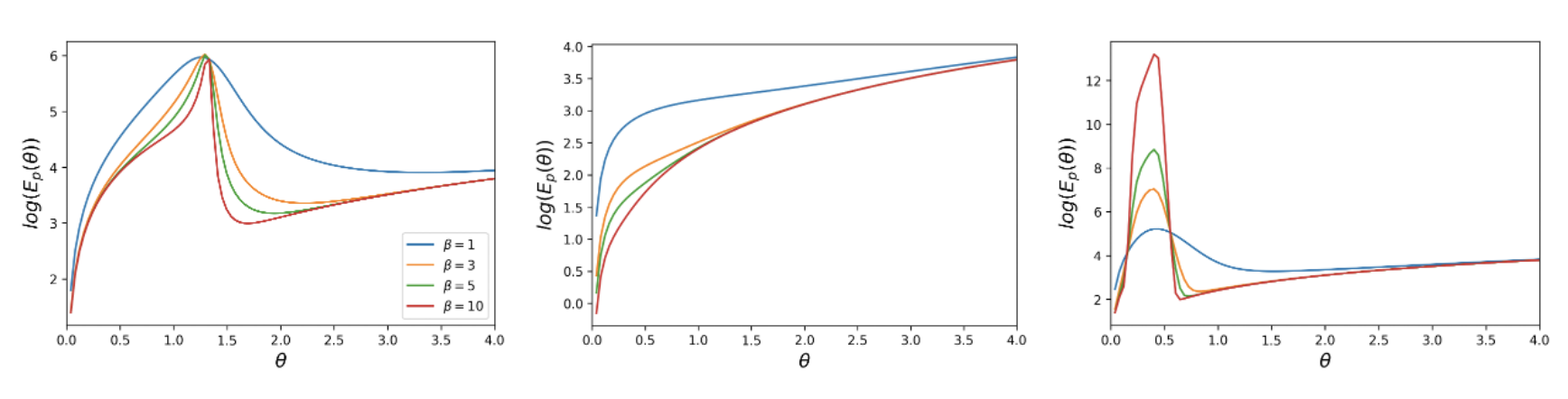}
\caption{Comparison of the punishment cost function $E_p(\theta)$ for the General $2 \times 2$ Games: Prisoner's Dilemma ($R = 2$, $S = 0$, $T = 4$, $P = 1$), Hawk and Dove ($R = 2$, $S = 1$, $T = 3$, $P = 0$), and Stag Hunt ($R = 3$, $S = 0$, $T = 2$, $P = 1$), presented in this order. The plots show the variation of $E_p(\theta)$ under different intensities of selection $\beta$, for a population size of $N = 10$.} 
\label{fig: punishment comparison}
\end{figure}

\begin{figure}[H]
\centering
\hspace*{-0.5cm} 
\includegraphics[width=\textwidth]{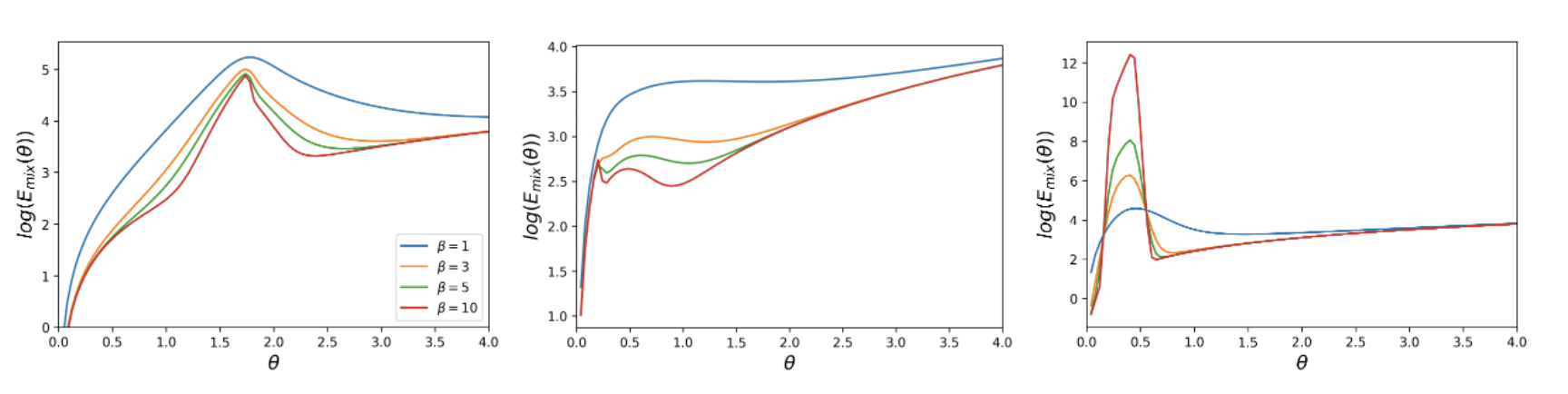}
\caption{Comparison of the hybrid cost function $E_{mix}(\theta)$ for the General $2 \times 2$ Games: Prisoner's Dilemma ($R = 2$, $S = 0$, $T = 4$, $P = 1$), Hawk and Dove ($R = 2$, $S = 1$, $T = 3$, $P = 0$), and Stag Hunt ($R = 3$, $S = 0$, $T = 2$, $P = 1$), in that order. The cost function is evaluated across various values of the intensity of selection $\beta$, with population size $N = 10$.} 
\label{fig: hybrid comparison}
\end{figure}

\begin{figure}[H]
\centering
\hspace*{-0.5cm}  
\includegraphics[width=\textwidth]{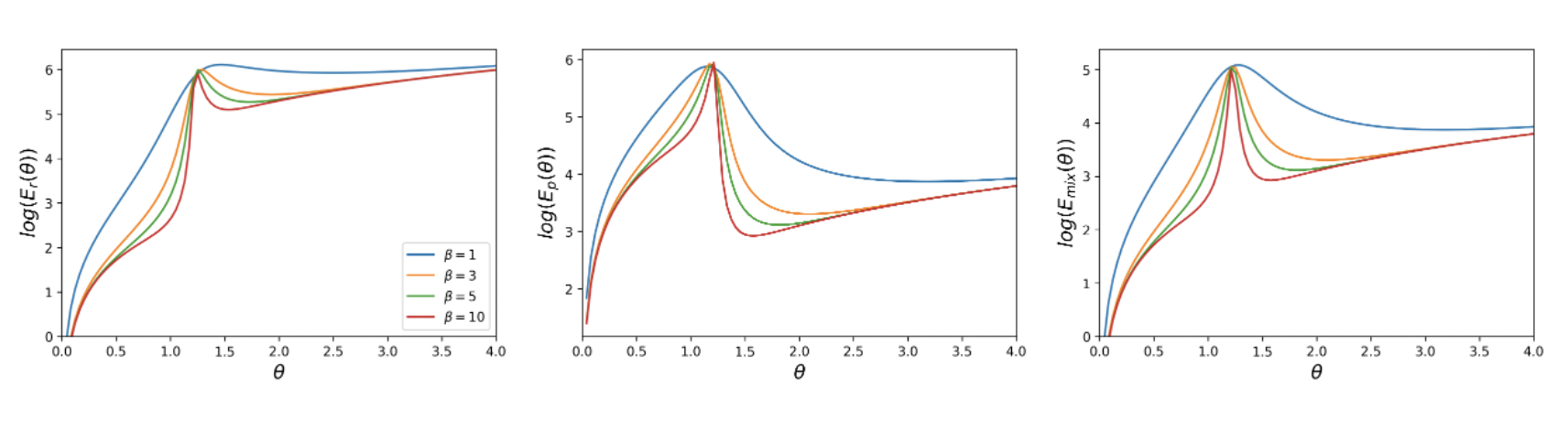}
\caption{Reward, punishment, and hybrid cost functions under varying values of the intensity of selection $\beta$ for the Donation Game, a special case of a $2 \times 2$ game satisfying $R + P = T + S$. This setting has been analysed in detail in \cite{duong2021cost, DuongDurbacHan2022, DuongDurbacHan2024}.} 
\label{fig: comparison DG}
\end{figure}

The first images in Figures \ref{fig: reward comparison}, \ref{fig: punishment comparison}, \ref{fig: hybrid comparison} illustrates the behaviour of the reward, punishment, and hybrid cost functions in the case of the Prisoner's Dilemma ($R+P\neq T+S$). We highlight the phase transition that emerges as $\beta$ increases: the cost functions initially rise, then decline, and subsequently rise again. This pattern aligns with prior analytical results for the Donation Game ($R+P=T+S$) obtained in \cite{duong2021cost,DuongDurbacHan2022,DuongDurbacHan2024}.
While these works assumed that the dynamics starts either in the state of all defectors $S_0$ or equally likely in the homogeneous states $S_0$ or $S_N$, herein we assume a general starting point. It should be noted that the phase transition behaviour persists even in this generalised setting. The increase in the cost per capita $\theta$ followed by a sharp decline as observed in the non-monotonicity of the cost function at higher intensities of selection could be explained by a large number of players switching to cooperation as a result of incentive use. 

In the Hawk and Dove game (second images in Figures \ref{fig: reward comparison}, \ref{fig: punishment comparison}, \ref{fig: hybrid comparison}), the reward cost function consistently increases without a transition (Figure \ref{fig: reward comparison} second image), while a phase transition can be observed for the punishment and the hybrid cost functions when the intensity of selection $\beta$ is large enough (Figures \ref{fig: punishment comparison}, \ref{fig: hybrid comparison}). Although for larger values of $\beta$ the cost function starts behaving non-monotonically, we note that the trend is, after the initial fluctuations, for the function to increase steadily. By contrast, in the Stag Hunt game (third images in Figures \ref{fig: reward comparison}, \ref{fig: punishment comparison}, \ref{fig: hybrid comparison}), we observe a bottleneck effect: as $\beta$ grows, costs initially rise to a peak before sharply decreasing and appearing to stabilise across all values of the intensity of selection (Figures \ref{fig: reward comparison}, \ref{fig: punishment comparison}, \ref{fig: hybrid comparison} third images). These findings highlight the influence of the game structure on the behaviour of incentive costs which seems to play a crucial role in determining whether the cost functions follow a monotonic or non-monotonic trajectory.

\subsection{Collective Risk Game}

\begin{figure}[H]
\centering
\includegraphics[width=\textwidth]{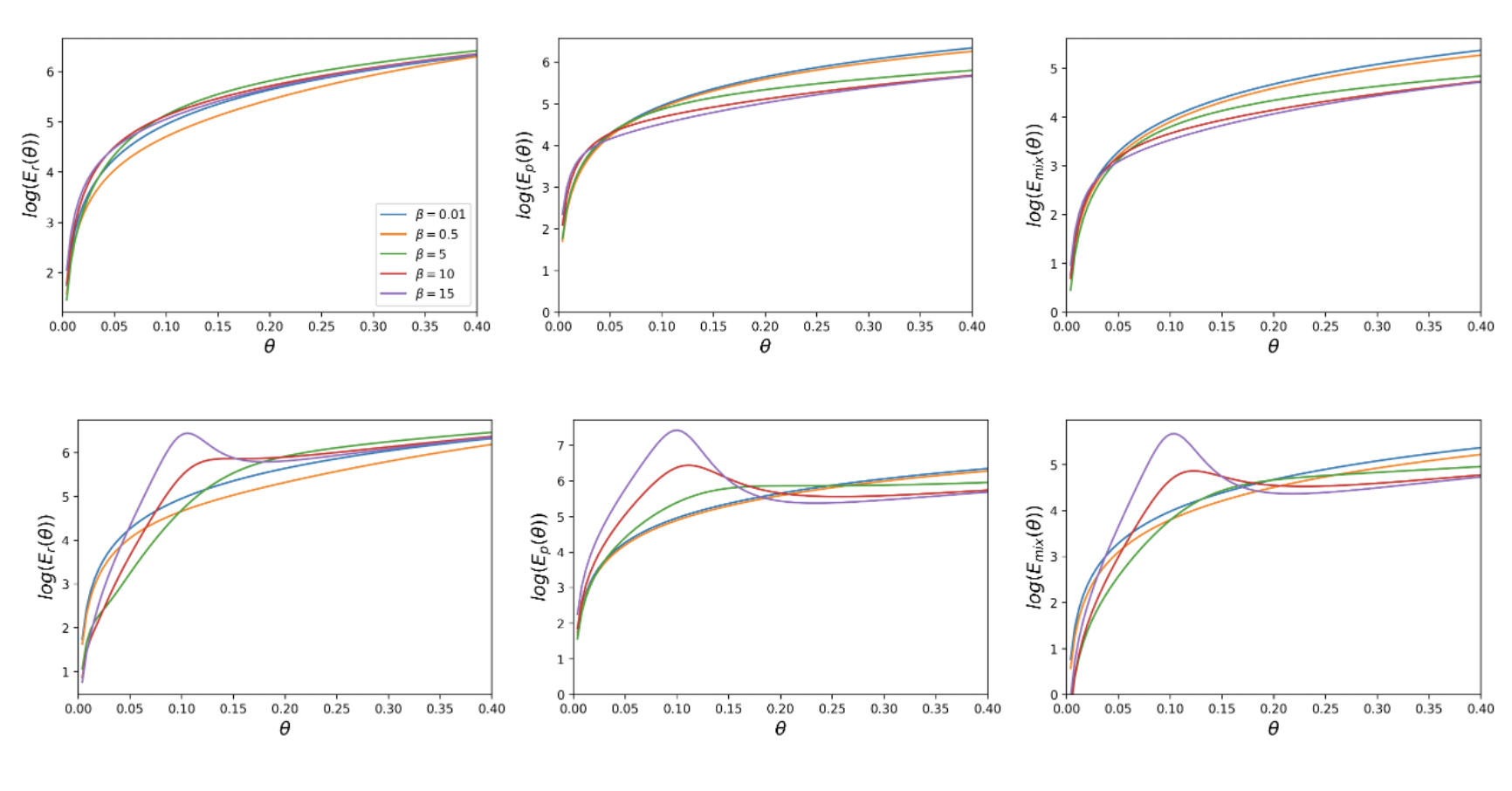}
\caption{Comparison of the cost functions for reward ($E_r(\theta)$; first column), punishment ($E_p(\theta)$; second column), and hybrid incentives ($E_{mix}(\theta)$; third column) in the Collective Risk Game with $B = 1$, $c = 0.1$, $n = 5$, $m = 2$, and population size $N = 20$. The first row corresponds to a lower risk level ($r = 0.3$) and the second row to a higher one ($r = 0.7$). Each figure shows the effect of varying the intensity of selection $\beta$ on the corresponding cost function.} 
\label{fig: CRG comparison}
\end{figure}

\begin{figure}[H]
\centering
\hspace*{-1cm}  
\includegraphics[width=0.7\textwidth]{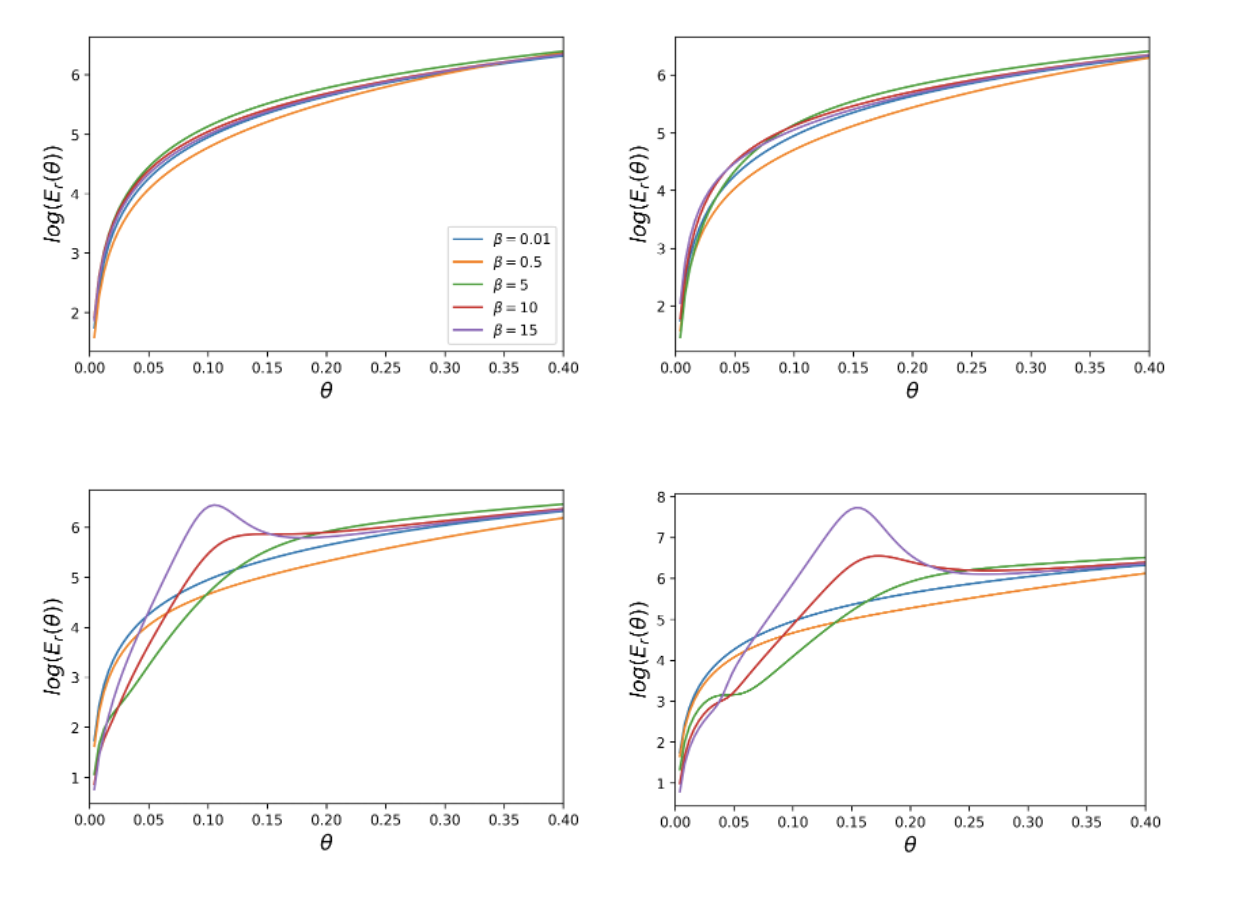}
\caption{Behaviour of the reward cost function $E_r(\theta)$ in the Collective Risk Game for varying intensities of selection $\beta$ and risk probabilities $r = 0.01$, $0.3$, $0.7$, and $0.9$ (in this order). Parameters are fixed at $B = 1$, $c = 0.1$, $n = 5$, $m = 2$, and $N = 20$.} 
\label{fig: comparison CRG reward}
\end{figure}

\begin{figure}[H]
\centering
\hspace*{-1cm}  
\includegraphics[width=0.7\textwidth]{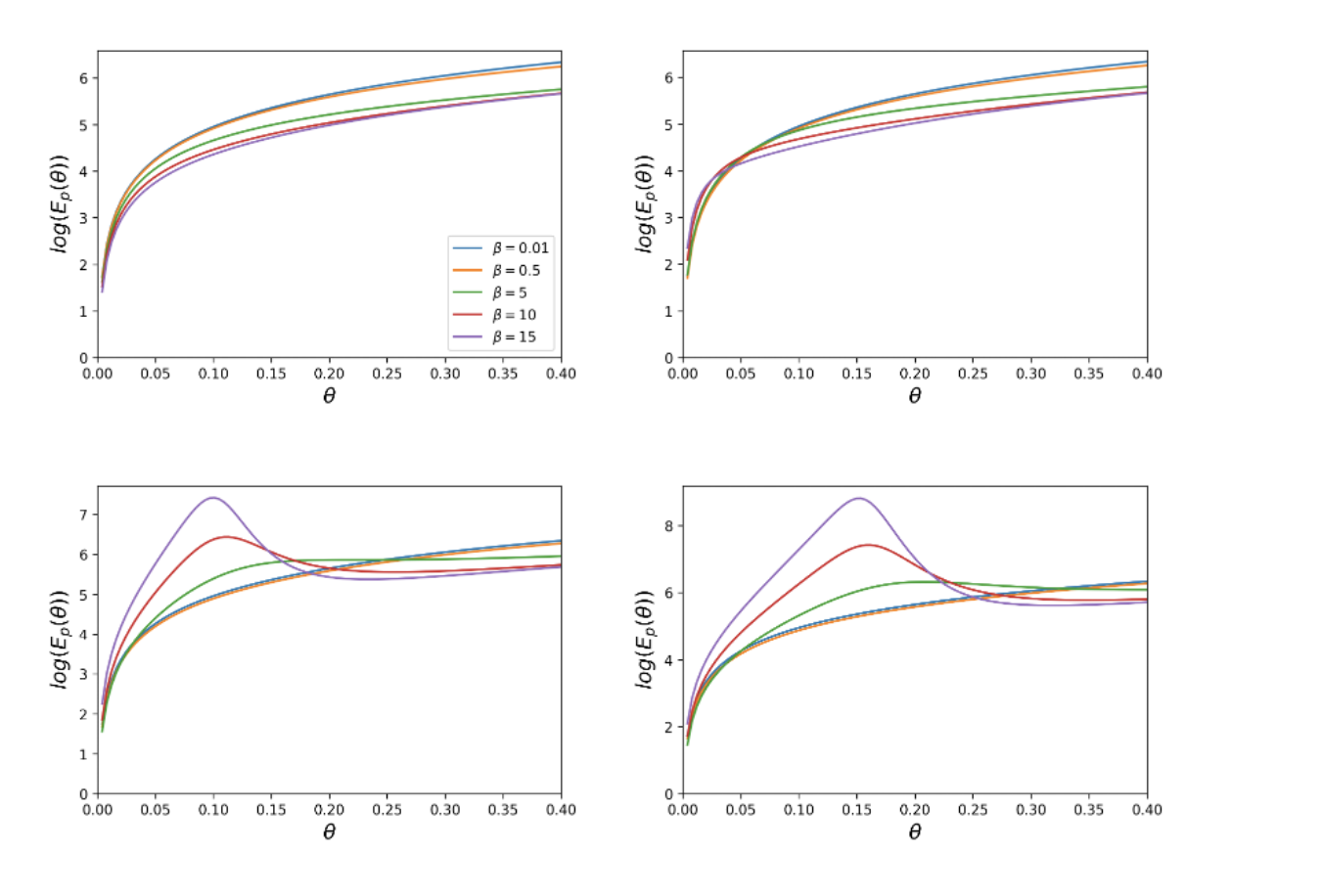}
\caption{Behaviour of the punishment cost function $E_p(\theta)$ in the Collective Risk Game for varying intensities of selection $\beta$ and risk probabilities $r = 0.01$, $0.3$, $0.7$, and $0.9$ (in this order). Parameters are fixed at $B = 1$, $c = 0.1$, $n = 5$, $m = 2$, and $N = 20$.} 
\label{fig: comparison CRG punishment}
\end{figure}

\begin{figure}[H]
\centering
\hspace*{-1cm}  
\includegraphics[width=0.7\textwidth]{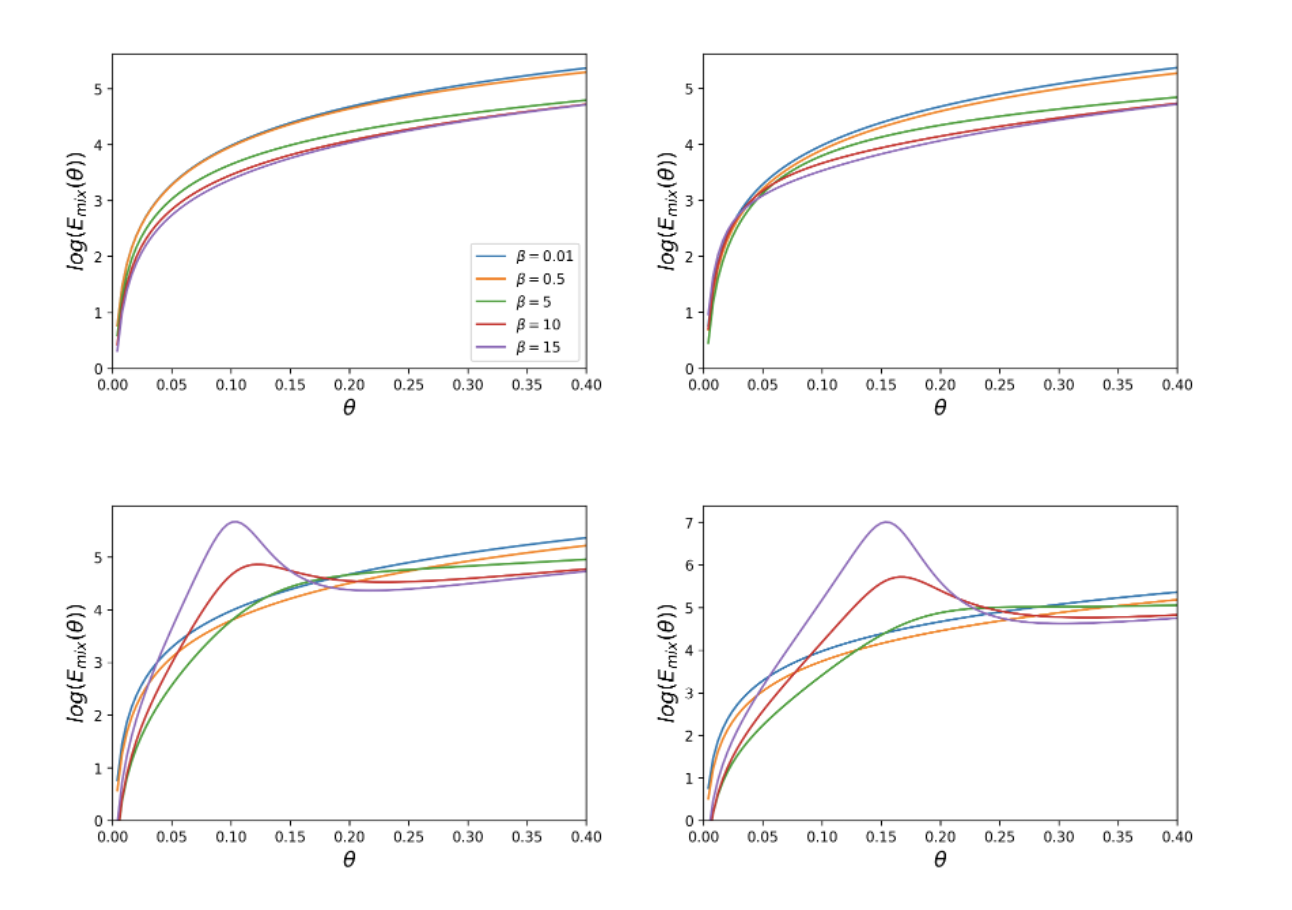}
\caption{Behaviour of the hybrid cost function $E_{mix}(\theta)$ in the Collective Risk Game for varying intensities of selection $\beta$ and risk probabilities $r = 0.01$, $0.3$, $0.7$, and $0.9$ (in this order). Parameters are fixed at $B = 1$, $c = 0.1$, $n = 5$, $m = 2$, and $N = 20$.} 
\label{fig: comparison CRG hybrid}
\end{figure}

The Collective Risk Game also exhibits a phase transition for larger values of the intensity of selection $\beta$ as seen in Figure \ref{fig: CRG comparison}. Additionally, the probability of losing the group endowment $r$ influences this dynamic: the higher 
the risk is, the more pronounced the phase transition becomes at larger $\beta$ as illustrated in Figures \ref{fig: comparison CRG reward}, \ref{fig: comparison CRG punishment}, \ref{fig: comparison CRG hybrid}. This suggests that both the intensity of selection $\beta$ and the probability of loss 
$r$ play a crucial role in shaping the behaviour of all incentive types considered in this work - reward, punishment, and hybrid.

The phase transitions observed in the cost functions arise from the interplay between the intensity of selection $\beta$, the incentive mechanisms, and the strategic decision-making under evolutionary dynamics. When $\beta$ is small, selection is weak, and strategies evolve more randomly, resulting in relatively stable costs. As $\beta$ increases, selection amplifies fitness differences, initially raising the cost of incentives as they become more necessary to support cooperation. At intermediate 
$\beta$, cooperation may stabilise, reducing the need for incentives and lowering costs. However, for very large $\beta$, selection becomes so strong that either cooperators dominate or defectors persist despite incentives, driving costs up again. In the Collective Risk Game, the probability of losing the group endowment $r$ further influences this behaviour. When $r$ is low, defectors may persist, but as $r$ increases, the threat of collective loss leads to a strengthening of cooperative behaviours, making the phase transition more pronounced at high $\beta$. 

The sensitivity of the institutional incentives to the underlying game also explains the differences observed across settings. In the Prisoner’s Dilemma, the strong temptation to defect results in non-monotonic cost behaviour, while in the Stag Hunt, coordination dynamics create a bottleneck effect, where incentives initially drive cooperation but become less necessary once a cooperative majority is reached. In contrast, the Hawk and Dove game, where mixed strategies dominate, leads to a consistently increasing reward cost function as incentives drive up competitive interactions. When punishment is introduced either alone or in combination with reward as is the case in hybrid incentives, we observe a phase transition similar to the one in the Prisoner's Dilemma, although, in the long-run, it seems that the function then tends to increase steadily. 

\subsection{Solution to the optimisation problem}

In this subsection, we present a solution to the minimisation problem displayed in \eqref{eq: min prob}. Our goal is to determine the optimal individual incentive cost $\theta^*$ that minimises the expected institutional cost $E(\theta)$ while achieving a target cooperation level $\omega$ within a population of size $N$.

Previously, we presented the behaviour of the cost function $E(\theta)$ in the case of reward, punishment, and hybrid incentives for general $2\times2$ games (in particular, the cooperative and defective Prisoner's Dilemma, Hawk and Dove, and Stag Hunt) as well as for the Collective Risk Game. The cost function $E(\theta)$ exhibits two behaviours - it is either monotonically increasing or non-monotonic, displaying a characteristic phase transition dependent on the intensity of selection.

When $E(\theta)$ is monotonically increasing, the optimal solution can be derived in closed-form from the fixation probabilities of the underlying Markov process. Specifically, the institution can minimise the expected cost by providing an incentive amount

$$
\theta_0(\omega) = \frac{1}{(N-1)\beta} \log\left(\frac{\omega}{1-\omega}\right) - \Delta,
$$ from Equation \eqref{eq:omega_fraction}.

However, if the cost function $E(\theta)$ exhibits a non-monotonic behaviour, the above solution no longer guarantees optimality. Due to the complexity of the cost function $E(\theta)$ which depends on inverting the fundamental matrix of the Markov chain, a rigorous mathematical solution is intractable. We therefore implement a robust numerical optimisation method in Python 3.10. The approach involves discretising the feasibility range of $\theta$, computing $E(\theta)$ via the fundamental matrix $W^{-1}$ of the Markov chain, and applying scalar minimisation methods such as the \textit{$minimize\_scalar$} function from SciPy. 

\begin{algorithm}
\begin{algorithmic}
\State \textbf{Input:} Population size $N$; game-specific parameters $R, S, T, P$ for general $2\times2$ games and $B, c, r, n, m$ for the Collective Risk Game; intensity of selection $\beta$; target cooperation level $\omega$
\State \textbf{Define:} Range of $\theta$ values, e.g. $\theta \in [0, 4]$
\State \textbf{Define:} Cost function $E(\theta)$ using the fundamental matrix of the Markov chain $W^{-1}$ and the definition in Equation~\eqref{eq: cost function with fC and fD}
\vspace{0.5em}
\Function{CostFunction}{$R, S, T, P, B, c, r, n, m, \beta, \theta, \omega$}
    \State Construct transition matrix $W$ using the Fermi strategy update rule
    \State Compute payoff difference $\Delta$
    \State Invert $W$ to obtain fundamental matrix $W^{-1}$
    \State Compute fixation probabilities $f_C$ and $f_D$
    \State Compute $E(\theta)$ with the information above
    \State \Return $E(\theta)$
\EndFunction
\vspace{0.5em}
\State Use numerical optimisation (\textit{minimize\_scalar} in SciPy) to compute:
\[
\theta^* = \arg\min_{\theta \in [0, 4]} E(\theta).
\]
\State \textbf{Output:} $\theta^*$ and $E(\theta^*)$ 
\end{algorithmic}
\end{algorithm}

\section{Discussion}
\label{sec: discussion}

This work addressed the problem of optimising the cost of institutional incentive schemes (reward, punishment, and hybrid) for promoting cooperation in finite, well-mixed populations in which individuals interact via a general $2\times2$ game or the Collective Risk Game, games in which the average payoff difference between cooperators and defectors depends on the population composition. We derived the cost functions of the institutional spending for providing reward and/or punishment, which we then analysed. Our results confirm and extend previous findings on the impact of the intensity of selection $\beta$ and the structural properties of the payoff matrices in determining the cost dynamics of institutional incentives.

One of the key findings in our study is the existence of a phase transition in the cost functions across various settings. In particular, we showed that, consistent with earlier work on the Donation Game ($R+P = T+S$) \cite{duong2021cost, DuongDurbacHan2022, DuongDurbacHan2024}, the cost functions for the Prisoner's Dilemma ($R+P \neq T+S$) increase initially with the intensity of selection, before exhibiting a non-monotonic pattern - a sharp decline, followed by a second rise. This behaviour exists in both the cooperative and defective versions of the game, suggesting that the observed phase transition is not an outlier, but a general feature of cooperation dilemmas under stochastic dynamics. Expanding our analysis, we then examined the effect of institutional incentives on other well-known games, such as Hawk and Dove and Stag Hunt, under varying intensities of selection. In the Hawk and Dove game, we observed a steady increase in the reward cost function and a phase transition in the cases of punishment and hybrid incentives. The Stag Hunt game introduced a bottleneck behaviour, where costs initially rose with the intensity of selection before sharply declining and converging to a stable value. This indicates that incentives in conflict dilemmas behave differently from those in cooperative-defective settings.

In the Collective Risk Game, our analysis revealed that both the intensity of selection $\beta$ and the probability of losing the endowment $r$ influence the phase transition of the cost functions. Higher values of $r$ amplify this transition effect, highlighting the idea that risk perception plays a vital role in shaping the effectiveness of incentives. These findings are in line with climate-related applications of the game, where the threat of loss has been shown to trigger a cooperative response \cite{gois2019reward, han2017evolution}.

These results suggest that policymakers who are designing institutional incentives should take into account the game parameters, the behavioural context together with the perceived risk of loss as well as the intensity of selection (an indicator of how strongly individuals tend to imitate each other). When the intensity is low (tending towards neutral drift), interventions are less likely to be effective, since there is no tendency to mimic strategies. In this case, institutional incentives might have little impact. By contrast, as the intensity of selection increases, institutional incentives are more likely to shape collective behaviour as the individuals tend to imitate each other more. A limitation of our work is the assumption of a well-mixed population, which might not fully capture the complexity of real-world dynamics. Moreover, we consider a full-incentive scheme (all individuals receive incentives), which may not be feasible in practice. We also assume that the external decision-maker has full knowledge of the population composition which might be costly in real life. 

Overall, our results provide a comprehensive understanding of the behaviour of the cost functions for reward, punishment, and hybrid incentive schemes in a finite, well-mixed population. By extending previous works, particularly by considering games where the payoff difference between cooperators and defectors depends on the state of the Markov chain, we contribute to the understanding of the conditions under which incentives are most effective. Future research could explore structured populations, other incentive mechanisms, or the combination of multiple mechanisms to study the dynamics of cooperation. 

 \renewcommand{\thefigure}{A\arabic{figure}}
 \renewcommand{\thetable}{A\arabic{table}}
 \setcounter{figure}{0}   
 
\section{Appendix}
\label{sec: appendix}

\subsection*{Transition probabilities for punishment incentives and hybrid incentives}
\label{sec: fundamental matrices punishment/hybrid}

In this subsection, we present the computation of the transition probabilities for punishment and hybrid incentives, which differ slightly from those computed for reward in Section \ref{sec: asymptotic limits}. Once the transition probabilities are obtained, the fundamental matrix of the Markov chain for these types of incentive is identical to the one for reward. 

\noindent For punishment, we have:
\begin{equation} 
\label{eq: transition probabilities punishment General 2x2 Games}
\begin{split} 
u_{i,i\pm k} &= 0 \qquad \text{ for all } k \geq 2, \\
u_{i,i\pm1} &= \frac{N-i}{N} \frac{i}{N} \left(1 + e^{\mp\beta[\delta_i +\theta_i/(N-i)]}\right)^{-1},\\
u_{i,i} &= 1 - u_{i,i+1} -u_{i,i-1}.
\end{split} 
\end{equation}
We normalise $b=1$ for simplicity and obtain (recalling that $\Pi_C(i)-\Pi_D (i)=\delta_i$ and that $\theta_i/(N-i)=\frac{\theta}{b}=\theta$ for $1\leq i\leq N-1$):
\begin{equation} 
\label{eq: transition probabilities punishment General 2x2 Games normalised}
\begin{split} 
u_{i,i\pm k} &= 0 \qquad \text{ for all } k \geq 2, \\
u_{i,i\pm1} &= \frac{N-i}{N} \frac{i}{N} \left(1 + e^{\mp\beta[\delta_i +\theta]}\right)^{-1},\\
u_{i,i} &= 1 - u_{i,i+1} -u_{i,i-1}.
\end{split} 
\end{equation}

We then proceed exactly the same as in the case of reward.\\

\noindent For hybrid incentives, we have:
\begin{equation} 
\label{eq: transition probabilities hybrid General 2x2 Games}
\begin{split} 
u_{i,i\pm k} &= 0 \qquad \text{ for all } k \geq 2, \\
u_{i,i\pm1} &= \frac{N-i}{N} \frac{i}{N} \left(1 + e^{\mp\beta[\delta_i +\theta_i/\min(i/a, (N-i)/b)]}\right)^{-1},\\
u_{i,i} &= 1 - u_{i,i+1} -u_{i,i-1}.
\end{split} 
\end{equation}
We normalise $a=b=1$ for simplicity and obtain (recalling that $\Pi_C(i)-\Pi_D (i)=\delta_i$ and that $\theta_i/\min(i/a, (N-i)/b)=\theta$ for $1\leq i\leq N-1$):
\begin{equation} 
\label{eq: transition probabilities hybrid General 2x2 Games normalised}
\begin{split} 
u_{i,i\pm k} &= 0 \qquad \text{ for all } k \geq 2, \\
u_{i,i\pm1} &= \frac{N-i}{N} \frac{i}{N} \left(1 + e^{\mp\beta[\delta_i +\theta]}\right)^{-1},\\
u_{i,i} &= 1 - u_{i,i+1} -u_{i,i-1}.
\end{split} 
\end{equation}

We then proceed exactly the same as in the case of reward.

\subsection*{Neutral drift limit for punishment and hybrid incentives}
\label{sec: appendix neutral drift}

\begin{lemma}(Neutral drift limit)
For both a general $2\times 2$ game and the Collective Risk Game, we have
$$\lim_{\beta\to 0}E(\theta)= \begin{cases}
    N^2\theta H_N \hspace{0.3cm} \textit{(Punishment)} \\
    N^2\theta H_{N,a,b} \hspace{0.3cm} \textit{(Hybrid)},
\end{cases}$$
where $H_N$ is the harmonic number
\[
H_N=\sum_{i=1}^{N-1}\frac{1}{j}
\]
and $H_{N,a,b}$ is defined by
\[
H_{N,a,b}=\sum_{k=1}^{N-1} \frac{1}{j(N-j)}\min\Big(\frac{j}{a},\frac{N-j}{b}\Big).
\]
\end{lemma}
\begin{proof}
\noindent It follows from the formula of the entries of the fundamental matrix $W$ given in \eqref{eq: matrix W General 2x2 Games} that, in the limit of neutral drift, i.e. when $\beta\rightarrow0$, we have $\lim_{\beta\rightarrow 0} W=\bar W$, where

\begin{align}
\bar W= \frac{1}{2}\begin{pmatrix} 
2&-1&&&&&\\
-1&2&-1&&&&&\\
&\ddots&\ddots&\ddots&&&\\
&&-1&2&-1&&&\\
&&&-1&2&-1&&&\\
&&&&\ddots&\ddots&\ddots&\\
&&&&&-1&2&-1\\
&&&&&&-1&2
\end{pmatrix}.
\end{align}
The tridiagonal matrix above is the well-known Cartan matrix. It follows from \cite{wei2017inverses} that $(\bar{W}^{-1})_{i,j}= 2\Big[\min(i,j)-\frac{ij}{N}\Big]$, for $1\leq i,j\leq N-1.$

Then, noting that  $\lim\limits_{\beta\to 0} f_C = \lim\limits_{\beta\to 0} f_D = \frac{1}{2}$, we have
\begin{equation}
\label{eq: neutral drift limit}
    \lim_{\beta\to 0}E(\theta)= \lim_{\beta\to 0}\sum_{j=1}^{N-1} (f_Dn_{1,j} + f_Cn_{N-1,j} )\theta_j = \frac{N^2}{2}\sum_{j=1}^{N-1} \frac{(\bar W^{-1}_{1,j} + \bar W^{-1}_{N-1,j})}{j(N-j)}\theta_j.
\end{equation}

Therefore, assuming without loss of generality that $b = 1$ and substituting $\theta_j=(N-j)\theta$ in Equation \eqref{eq: neutral drift limit}, we get 

$$
 \lim_{\beta\to 0}E_p(\theta) = N^2\theta H_N.
$$

In the case of hybrid incentives, where $\theta_j=\theta\min(\frac{j}{a},\frac{N-j}{b})$, and assuming without loss of generality that $a=b=1$, we have

\begin{align*}
    \lim_{\beta\to 0}E_{mix}(\theta) &= \frac{N^2}{2}\Big[\sum_{j=1}^{N-1} \frac{\bar W^{-1}_{1,j} + \bar W^{-1}_{N-1,j}}{N-j} \theta_j\Big] \\
    &= \frac{N^2}{2}\Big[\sum_{j=1}^{N-1} \frac{2(N-j)}{jN(N-j)} \theta \min(\frac{j}{a},\frac{N-j}{b}) + \sum_{j=1}^{N-1} \frac{2j}{jN(N-j)} \theta \min(\frac{j}{a},\frac{N-j}{b})\Big] \\
    &= N\theta \Big[ \sum_{j=1}^{N-1} \frac{1}{j}\min(\frac{j}{a},\frac{N-j}{b}) + \sum_{j=1}^{N-1} \frac{1}{N-j}\min(\frac{j}{a},\frac{N-j}{b}) \Big] \\
    &= N\theta \Big[ \sum_{k=1}^{N-1} \Big(\frac{1}{j}+\frac{1}{N-j}\Big)\min(\frac{j}{a},\frac{N-j}{b}) \Big] \\
    &= N^2\theta\Big[ \sum_{k=1}^{N-1} (\frac{1}{j(N-j)})\min(\frac{j}{a},\frac{N-j}{b}) \Big] \\
    &= N^2 \theta H_{N,a,b}.
\end{align*}
\end{proof}

\begin{figure}[H]
\centering
\hspace*{-0.5cm}  
\includegraphics[width=\textwidth]{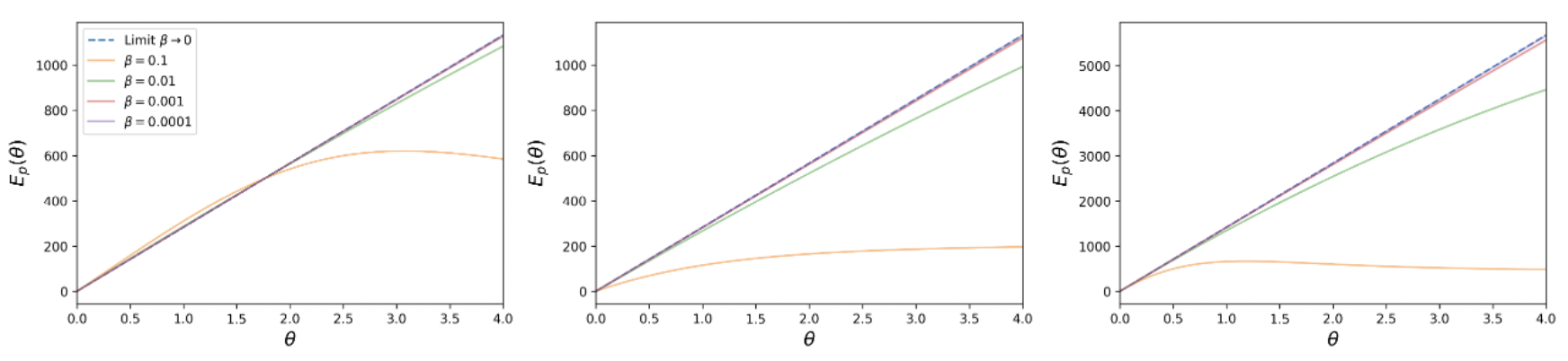}
\caption{The dashed blue line represents the neutral drift limit of the punishment cost function. The three figures correspond to: (i) a defective Prisoner's Dilemma (left; $R = 2$, $S = 0$, $T = 4$, $P = 1$; restricted such that $\delta + \theta < 0$), (ii) a cooperative Prisoner's Dilemma (centre; $R = 1$, $S = 4$, $T = 0$, $P = 2$), and (iii) the Collective Risk Game (right; $B = 1$, $c = 0.1$, $r = 0.9$, $n = 5$, $m = 2$). As the intensity of selection $\beta$ decreases, the punishment cost function converges towards the neutral drift limit.} 
\label{fig: neutral drift limit PD}
\end{figure}

\begin{figure}[H]
\centering
\hspace*{-0.5cm}  
\includegraphics[width=\textwidth]{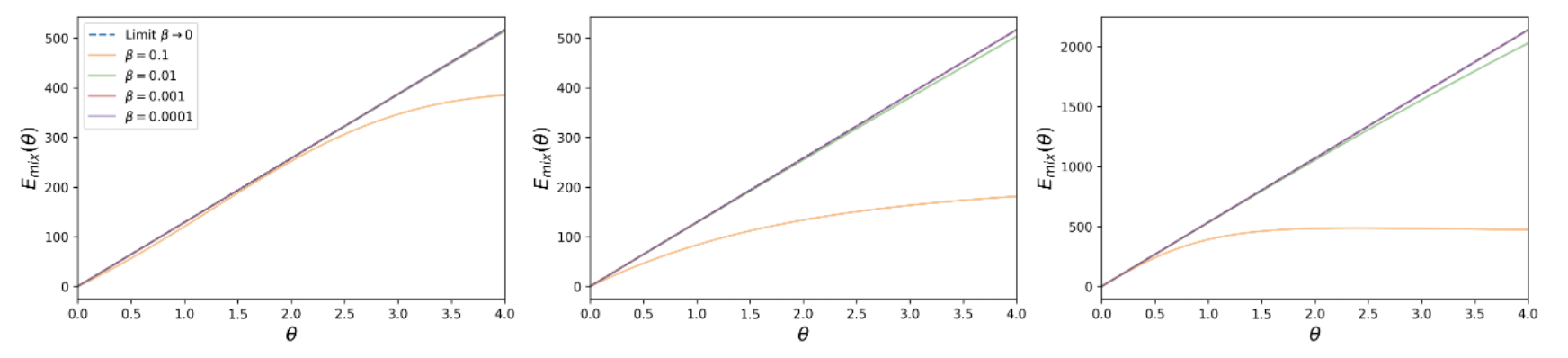}
\caption{The dashed blue line represents the neutral drift limit of the hybrid cost function. The three figures correspond to: (i) a defective Prisoner's Dilemma (left; $R = 2$, $S = 0$, $T = 4$, $P = 1$; restricted such that $\delta + \theta < 0$), (ii) a cooperative Prisoner's Dilemma (centre; $R = 1$, $S = 4$, $T = 0$, $P = 2$), and (iii) the Collective Risk Game (right; $B = 1$, $c = 0.1$, $r = 0.9$, $n = 5$, $m = 2$). As the intensity of selection $\beta$ decreases, the hybrid cost function converges towards the neutral drift limit.} 
\label{fig: neutral drift limit PD}
\end{figure}

\subsection*{Strong selection limit for punishment and hybrid incentives}

\begin{lemma}(Strong selection limit)
\begin{enumerate}
    \item \label{lem: ssl part 1 appendix} Defective Prisoner's Dilemma - Recall that, for the defective Prisoner's Dilemma, $\delta_j < 0$ for all $j$. Assume that $\theta$ is sufficiently small so that $\delta_j + \theta < 0$ for all $j$. Then, in the limit of strong selection, 
\[
\lim_{\beta \to \infty} E(\theta) =
\begin{cases}
\displaystyle N^2 \theta & \text{(Punishment)} \\
\displaystyle \frac{N^2 \theta}{a(N-1)} & \text{(Hybrid)}.
\end{cases}
\]
\item \label{lem: ssl part 2 appendix} Cooperative and defective Prisoner's Dilemma, the Collective Risk Game - Recall that, for the cooperative Prisoner's Dilemma, $\delta_j > 0$ for all $j$. Assume that $\theta$ is sufficiently large so that $\delta_j + \theta > 0$ for all $j$ also holds in the defective Prisoner's Dilemma and in the Collective Risk Game. Then, in the limit of strong selection,
\[
\lim_{\beta \to \infty} E(\theta) =
\begin{cases}
\displaystyle \frac{N^2 \theta}{N-1}  & \text{(Punishment)} \\
\displaystyle \frac{N^2 \theta}{b(N-1)} & \text{(Hybrid)}.
\end{cases}
\]
\end{enumerate}
\end{lemma}

\begin{proof}
The fundamental matrix $W$ transforms as following in the limit of strong selection, i.e. when $\beta\rightarrow\infty$, in part \ref{lem: ssl part 1 appendix} and part \ref{lem: ssl part 2 appendix} of the lemma:
\begin{equation}
\label{eq: matrix W in beta infy limit for delta <= 0 appendix}
\hat W=\begin{pmatrix}
1&0&&&&&\\
-1&1&0&&&&&\\
&\ddots&\ddots&\ddots&&&\\
&&-1&1&0&&&\\
&&&-1&1&0&&&\\
&&&&\ddots&\ddots&\ddots&\\
&&&&&-1&1&0\\
&&&&&&-1&1
\end{pmatrix},
\end{equation}
and
\begin{equation}
\label{eq: matrix W in beta infy limit for delta > 0 appendix}
\overline{W}=\begin{pmatrix}
1&-1&&&&&\\
0&1&-1&&&&&\\
&\ddots&\ddots&\ddots&&&\\
&&0&1&-1&&&\\
&&&0&1&-1&&&\\
&&&&\ddots&\ddots&\ddots&\\
&&&&&0&1&-1\\
&&&&&&0&1
\end{pmatrix}.
\end{equation}

Using Lemma \ref{lem: Toeplitz inverse}, we can compute the matrix inverse for both \eqref{eq: matrix W in beta infy limit for delta <= 0 appendix} and \eqref{eq: matrix W in beta infy limit for delta > 0 appendix} as follows.\\

For \eqref{eq: matrix W in beta infy limit for delta <= 0 appendix}, $\lambda_+=1$ and $\lambda_-=0$, thus
\begin{equation*}
\hat W^{-1}_{i,j} = 
\begin{cases}
     0 \hspace{0.5cm} i<j \quad \\
     1 \hspace{0.5cm} i\geq j 
\end{cases}    
\end{equation*}

Therefore, $\hat W^{-1} =\begin{pmatrix}
1      & 0      & 0      & \cdots & 0 \\
1      & 1      & 0      & \cdots & 0 \\
1      & 1      & 1      & \cdots & 0 \\
\vdots & \vdots & \vdots & \ddots & \vdots \\
1      & 1      & 1      & \cdots & 1 \\
\end{pmatrix}.$\\

For \eqref{eq: matrix W in beta infy limit for delta > 0 appendix}, $\lambda_+=1$ and $\lambda_-=0$, thus
\begin{equation*}
\overline{W}^{-1}_{i,j} = 
\begin{cases}
     1 \hspace{0.5cm} i\leq j \quad \\
     0 \hspace{0.5cm} i> j 
\end{cases}    
\end{equation*}

Therefore, $\overline{W}^{-1} =\begin{pmatrix}
1 & 1 & 1 & \cdots & 1 \\
0 & 1 & 1 & \cdots & 1 \\
0 & 0 & 1 & \cdots & 1 \\
\vdots & \vdots & \vdots & \ddots & \vdots \\
0 & 0 & 0 & \cdots & 1 \\
\end{pmatrix}.$\\

In the limit of strong selection, $f_D = \frac{1}{e^{\beta (N-1)(\Delta + \theta)}+1}$ and $f_C = \frac{e^{\beta (N-1)(\Delta +  \theta)}}{e^{\beta (N-1)(\Delta +  \theta)}+1}$ transform as follows.

Recall $\Delta=
\frac{1}{N-1}\sum_{j=1}^{N-1}\delta_j$. Then
$$
\Delta+\theta= \frac{1}{N-1}\sum_{j=1}^{N-1}\delta_j + \theta = \frac{1}{N-1}\sum_{j=1}^{N-1}\Big(\delta_j+\theta\Big).
$$

In part \ref{lem: ssl part 1 appendix}, as $\delta_j+\theta<0$, we have $\Delta+\theta<0$. Similarly, in part \ref{lem: ssl part 2 appendix}, since $\delta_j+\theta>0$, we get that $\Delta+\theta>0$. Thus,

\[
\lim_{\beta \to \infty} f_C =
\begin{cases}
0 & \text{for } \text{part \ref{lem: ssl part 1 appendix}}\\
1 & \text{for } \text{part \ref{lem: ssl part 2 appendix}}
\end{cases}
\quad
\lim_{\beta \to \infty} f_D =
\begin{cases}
1 & \text{for } \text{part \ref{lem: ssl part 1 appendix}} \\
0 & \text{for } \text{part \ref{lem: ssl part 2 appendix}}.
\end{cases}
\]

Therefore, by rewriting the cost function $E(\theta)$ from Equation \eqref{eq: cost function with fC and fD}, we get:

\[
\lim_{\beta \to \infty} E(\theta) =
\begin{cases}
N^2 \sum_{j=1}^{N-1} \frac{\overline{W}_{N-1,j}}{j(N-j)}\theta_j & \text{for part \ref{lem: ssl part 1 appendix}} \\
N^2 \sum_{j=1}^{N-1} \frac{\hat W_{1,j}}{j(N-j)}\theta_j & \text{for part \ref{lem: ssl part 2 appendix}} .
\end{cases}
\]

Substituting the values of $\hat W_{i,j}$ and $\overline{W} _{i,j}$ computed in \eqref{eq: matrix W in beta infy limit for delta <= 0 appendix} and \eqref{eq: matrix W in beta infy limit for delta > 0 appendix}, respectively, as well as the definition of $\theta_j$ from Equation \eqref{eq: incentives per generation}, we get 

\[
\lim_{\beta \to \infty} E(\theta) =
\begin{cases}
\displaystyle
\begin{aligned}
&\text{part \ref{lem: ssl part 1 appendix}:}
\begin{cases}
N^2 \theta & \text{(Punishment)} \\
\displaystyle \frac{N^2 \theta}{a(N-1)} & \text{(Hybrid)}
\end{cases}\\[3ex]
&\text{part \ref{lem: ssl part 2 appendix}:}
\begin{cases}
\displaystyle \frac{N^2 \theta}{N-1} & \text{(Punishment)} \\
\displaystyle \frac{N^2 \theta}{b(N-1)} & \text{(Hybrid)}.
\end{cases}
\end{aligned}
\end{cases}
\]
\end{proof}

\begin{figure}[H]
\centering
\hspace*{-0.5cm}  
\includegraphics[width=\textwidth]{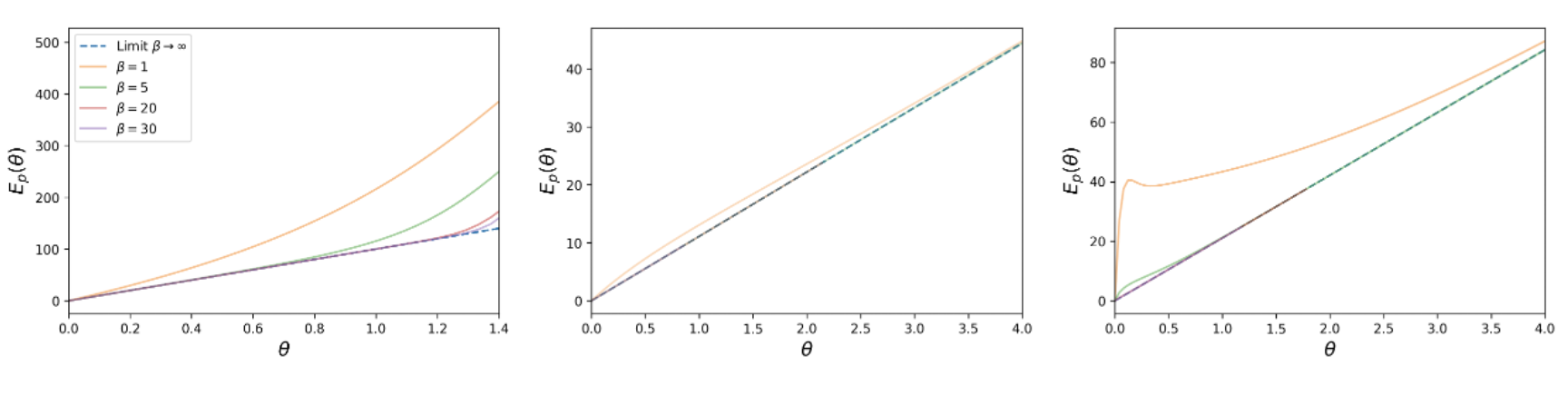}
\caption{The dashed blue line represents the strong selection limit of the punishment cost function. The three figures correspond to: (i) a defective Prisoner's Dilemma (left; $R = 2$, $S = 0$, $T = 4$, $P = 1$; restricted such that $\delta + \theta < 0$), (ii) a cooperative Prisoner's Dilemma (centre; $R = 1$, $S = 4$, $T = 0$, $P = 2$), and (iii) the Collective Risk Game (right; $B = 1$, $c = 0.1$, $r = 0.9$, $n = 5$, $m = 2$; restricted such that $\delta + \theta > 0$). As the intensity of selection $\beta$ increases, the punishment cost function asymptotically approaches the strong selection limit.} 
\label{fig: strong selection limit PD}
\end{figure}

\begin{figure}[H]
\centering
\hspace*{-0.5cm}  
\includegraphics[width=\textwidth]{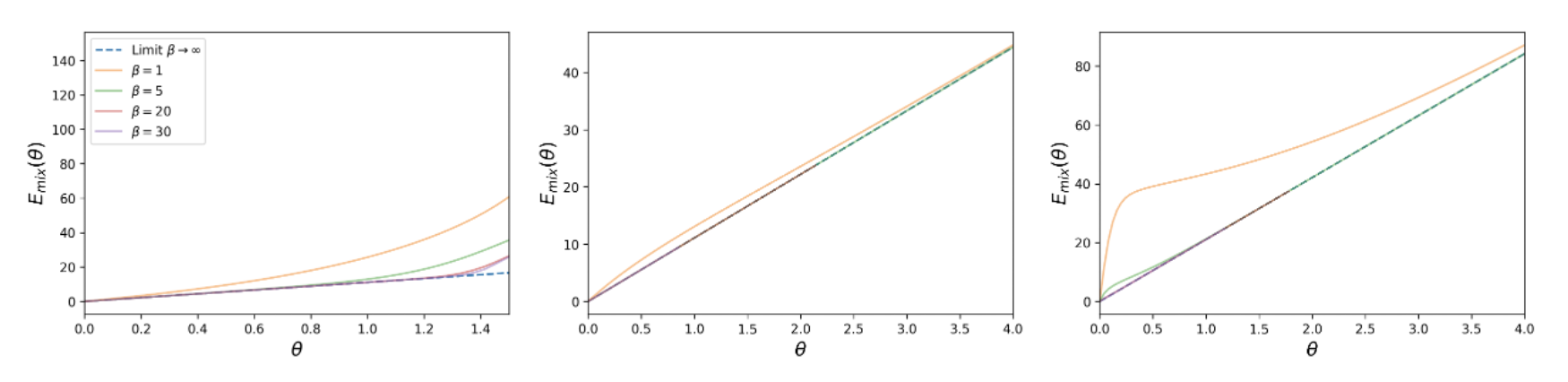}
\caption{The dashed blue line represents the strong selection limit of the hybrid cost function. The three figures correspond to: (i) a defective Prisoner's Dilemma (left; $R = 2$, $S = 0$, $T = 4$, $P = 1$; restricted such that $\delta + \theta < 0$), shown for the regime $\delta_j + \theta < 0$ for tractability; (ii) a cooperative Prisoner's Dilemma (centre; $R = 1$, $S = 4$, $T = 0$, $P = 2$); and (iii) the Collective Risk Game (right; $B = 1$, $c = 0.1$, $r = 0.9$, $n = 5$, $m = 2$; restricted such that $\delta + \theta > 0$). As the intensity of selection $\beta$ increases, the hybrid cost function converges towards the strong selection limit.} 
\label{fig: strong selection limit PD}
\end{figure}

\section*{Acknowledgement} The research of M.H.D. was supported by EPSRC Grants EP/W008041/1 and EP/V038516/1. C.M.D is supported by an EPSRC Studentship. T.A.H. is supported by EPSRC (grant EP/Y00857X/1). 
\section*{Data Availability Statement}
Data sharing is not applicable to this article as no datasets were generated or analysed during the current study.
\bibliographystyle{plain}
\bibliography{references}
\end{document}